%% file: 0main.tex
\documentclass[acmsmall,screen,authorversion,nonacm]{acmart}
\usepackage{booktabs} 
\usepackage[ruled]{algorithm2e} 

\SetAlFnt{\small}
\SetAlCapFnt{\small}
\SetAlCapNameFnt{\small}
\SetAlCapHSkip{0pt}
\IncMargin{-\parindent}


\usepackage{tikz}
\usetikzlibrary{automata}
\usetikzlibrary{calc, automata, chains, arrows.meta}
\usetikzlibrary{chains,scopes}
\usetikzlibrary{positioning}
\usepackage{multirow,subfig}

\newtheorem{theorem}{Theorem}
\newtheorem{lemma}[theorem]{Lemma}
\newtheorem{corollary}[theorem]{Corollary}

\theoremstyle{definition}
\newtheorem{definition}{Definition}[section]

\newcommand{\bR}{\mathbb{R}}
\newcommand{\bZ}{\mathbb{Z}}
\newcommand{\cD}{\mathcal{D}}

\newcommand{\hepsilon}{\widehat{\epsilon}}

\DeclareMathOperator*{\argmax}{arg\,max}

\DeclareMathOperator{\Exp}{\mathbb{E}}
\DeclareMathOperator{\Prob}{\mathbb{P}}

\newcommand{\unif}{\texttt{unif}}
\newcommand{\round}{\texttt{round}}
\newcommand{\lex}{\texttt{lex}}
\newcommand{\worst}{\texttt{worst}}
\newcommand{\best}{\texttt{best}}

\title[Best-Response Dynamics in Lottery Contests]{Best-Response Dynamics in Lottery Contests}

\author{Abheek Ghosh}
\email{abheek.ghosh@cs.ox.ac.uk}
\orcid{0000-0002-4771-4612}
\author{Paul W. Goldberg}
\email{paul.goldberg@@cs.ox.ac.uk}
\orcid{0000-0002-5436-7890}
\affiliation{%
    \department{Department of Computer Science}
    \institution{University of Oxford}
    \city{Oxford}
    \country{UK}
}

\begin{abstract}
We study the convergence of best-response dynamics in lottery contests. We show that best-response dynamics rapidly converges to the (unique) equilibrium for homogeneous agents but may not converge for non-homogeneous agents, even for two non-homogeneous agents. For $2$ homogeneous agents, we show convergence to an $\epsilon$-approximate equilibrium in $\Theta(\log\log(1/\epsilon))$ steps. For $n \ge 3$ agents, the dynamics is not unique because at each step $n-1 \ge 2$ agents can make non-trivial moves. We consider a model where the agent making the move is randomly selected at each time step. We show convergence to an $\epsilon$-approximate equilibrium in $O(\beta \log(n/(\epsilon\delta)))$ steps with probability $1-\delta$, where $\beta$ is a parameter of the agent selection process, e.g., $\beta = n$ if agents are selected uniformly at random at each time step. Our simulations indicate that this bound is tight.
\end{abstract}

\begin{document}


\maketitle


\input{1intro.tex}
\input{2prelim.tex}

\input{3non-hom.tex}
\input{4hom-2.tex}
\input{5hom-n-ub.tex}
\input{7simulation.tex}
\input{8conclude.tex}

\bibliographystyle{ACM-Reference-Format}
\bibliography{ref}

\appendix
\input{9appendix.tex}


\end{document}

%% file: 1intro.tex
\section{Introduction}\label{sec:into}
Contests are games where agents compete by making costly and irreversible investments to win valuable prizes. The lottery contest by \citet{tullock1980efficient} is one of the most widely used models for studying these environments and has been applied to problems in economics, political science, and computer science~\cite{konrad2009strategy,vojnovic2015contest}. 
To give a concrete real-life application, consider the game among miners in proof-of-work cryptocurrencies like Bitcoin~\cite{chen2019axiomatic,leshno2020bitcoin}. A Bitcoin miner adds the next block (and collects the corresponding reward) with a probability proportional to her (costly) computational power. A lottery contest exactly models the game among these miners (given they are risk neutral): a set of agents compete to win a valuable prize by investing costly effort, and an agent gets a portion of the prize proportional to his effort.

We study the convergence of best-response dynamics in lottery contests. The existence of an equilibrium may not always be a good predictor of the agents' behavior. The traditional explanation of equilibrium is that it results from analysis and introspection by the agents in a situation where the rules of the game, the rationality of the agents, and the agents’ payoff functions are all common knowledge. Both conceptually and empirically, these theories may have problems~\cite{fudenberg1998theory}. A model of the outcome of a game is more plausible if it can be attained in a decentralized manner, ideally via a process that involves agents behaving in a natural-looking self-interested way. Best-response dynamics is arguably the simplest of these models. In best-response dynamics, agents sequentially update their current strategy with one that best responds to that of the other agents. It is especially appropriate in settings, such as lottery contests, where pure-strategy equilibria are guaranteed to exist.

In our model for best response dynamics, we assume that at each time step one of the agents best responds and updates her current strategy. For two agents, this leads to a deterministic process where the agents play alternate moves.\footnote{Two consecutive moves by the same agent is redundant because the agent is already best responding after the first move and the second move does not change anything.} For more than two agents, the dynamics is not fully defined because all agents except the agent who made the most recent move can make a non-trivial transition. So, we consider a randomized model where an agent $i$ best responds with probability $w_{t,i}(x_t)$ at time $t$, given the strategy profile at time $t$ is $x_t$.\footnote{We make the technical assumption that $w_{t,i}(x_t) > 0$ for all agents except the one who played the most recent move and $w_{t,i}(x_t) < 1/2$ for all agents. See Section~\ref{sec:prelim:random} for more details.} We analyze how quickly best-response dynamics converges to an $\epsilon$-approximate equilibrium.

\subsection{Our Results}\label{sec:results}
For two homogeneous agents, we show that best response dynamics converges to an $\epsilon$-approximate pure-strategy Nash equilibrium in $\Theta(\log\log(1/\epsilon) + \log\log(1/\gamma))$ steps, where $\gamma$ is a reasonable function of the initial action profile (generally equal to the smallest positive action in the initial profile). On the other hand, for $n \ge 3$ homogeneous agents, we show convergence in $O(\alpha \log(n/(\epsilon\delta)) + \beta \log\log(1/\gamma))$ steps with probability $1-\delta$, where $\alpha$ and $\beta$ are related to the randomized agent selection process (see Theorem~\ref{thm:hom}), e.g., $\alpha = n$ and $\beta = 1$ if agents are selected uniformly at random each time step. 
We also provide a lower bound of $\Omega(\eta \log(n\delta) + \log\log(1/\epsilon) + \log\log(1/\gamma))$ with probability $1-\delta$, where $\eta$ is again related to the randomized agent selection process, e.g., $\eta = n$ for uniform selection.
For non-homogeneous agents, we provide examples to show that the dynamics may not converge, even for $2$ agents. 
We finish with some simulations that support our results.

\paragraph{Intuition for the proof of convergence.}
We divide the analysis into three main parts. 
The first part corresponds to a \textit{warm-up} phase of the dynamics. Intuitively, there is a certain domain $\cD$ such that the best response dynamics avoids corner cases and is easier to analyze if the output profile is in $\cD$. The warm-up phase corresponds to the time it takes to ensure that the output profile inside $\cD$ (where it stays thereafter) and is proportional to the time it takes for every agent to make at least one move. 

The third part of our analysis requires the total output to be sufficiently large often. In the second part of the analysis, we prove this result. We discretize the range of values that the total output can take into countably infinitely many intervals. We consider the discrete stochastic process that indicates the interval in which the total output lies at a given time point. We analyze this discrete stochastic process and prove that it frequently visits the states that correspond to a high total output, as required. We do this using suitable concentration bounds and coupling in Markov chains~\cite{mitzenmacher2017probability}. This part of the analysis also uses a few interesting insights related to the best response dynamics, e.g., the total output never decreases by more than a constant factor in a time step after the warm-up phase.

In the third part of the analysis, we use the potential function $f$ given in equation \eqref{eq:potential} to prove the rate of convergence. We show that the potential $f$ is strongly convex and smooth when the total output is sufficiently large, which we then use to show that the potential decreases by a multiplicative factor every time step. Our analysis here is motivated by the $\log(1/\epsilon)$ convergence rate achieved for the minimization of strongly convex and smooth functions using coordinate descent~\cite{wright2015coordinate,bertsekas2015convex,boyd2004convex}. On the other hand, the potential function is always non-increasing. Putting these results together, we get our rate of convergence.

Notice that our convergence rate is proportional to $\log(1/\epsilon)$ rather than $1/\epsilon$. The potential function $f$ can be used to give a $1/\epsilon$ convergence rate with less effort, but 
a $1/\epsilon$ convergence rate only gives a pseudo-polynomial time convergence rather than the polynomial time convergence of $\log(1/\epsilon)$.\footnote{If the homogeneous cost $c$ is small, i.e., $1/c$ is a large integer, then an additive $1$-approximate equilibrium may take $O(1/c)$ steps given a $1/\epsilon$ bound, which is exponential in the input size of $\log(1/c)$.} Many hard complexity classes, like PLS~\cite{johnson1988easy} and CLS~\cite{fearnley2022complexity}, always have $1/\epsilon$ pseudo-polynomial time convergence, and such a convergence rate is generally less interesting from a complexity theory perspective.


\subsection{Related Work}
The paper by \citet{ewerhart2017lottery} is the most closely related to ours. \citet{ewerhart2017lottery} showed that a lottery contest with homogeneous agents is a best-response potential game. The set of best-response potential games~\cite{voorneveld2000best,kukushkin2004best,dubey2006strategic,uno2007nested,uno2011strategic,park2015potential} is a super-set of ordinal potential games, which itself is a super-set of exact potential games~\cite{monderer1996potential}. In an exact potential game, the change in the potential function matches exactly the change in the utility of the agent making the move, while in an ordinal potential game, the change in the potential function has the same sign as the change in the utility of the agent making the move. A best-response potential game satisfies a weaker requirement: given the profile $x_{-i}$ of all agents except $i$, the best-response of agent $i$ matches the action that maximizes the potential function given $x_{-i}$. Notice that being a best-response potential game does not imply anything about actions that are not best responses, i.e., actions that increase an agent's utility but are not best responses may decrease the potential function and vice-versa. Lottery contests do not have an exact potential, so they are less structured than exact potential games and are strictly outside the class~\cite{ewerhart2017lottery}. Lottery contests with non-homogeneous agents have also been proven not to be ordinal potential games~\cite{ewerhart2020ordinal}.

In a best-response potential game, every non-trivial best-response transition increases the potential function. For a bounded potential function, it is not hard to see that being a best-response potential game guarantees the convergence of best-response dynamics, given that there is always a best response to every action profile. Some games, including the lottery contest, may not always have a best response. In particular, as discussed in Section~\ref{sec:prelim:br}, an agent does not have a best response in the lottery contest if every other agent is playing $0$. One of the early steps in our convergence proof shows that such corner cases do not occur too many times.

Exact potential games are equivalent to explicit congestion games. The complexity of computing a pure-strategy Nash equilibrium in these games is PLS-complete~\cite{johnson1988easy,fabrikant2004complexity} and a mixed-strategy Nash equilibrium is CLS-complete~\cite{fearnley2022complexity,babichenko2021settling}; these classes are believed to be hard (like NP-complete problems are believed to be hard, although arguably less strongly). Moreover, these hardness results imply that convergence of best-response dynamics can take exponentially many steps in the worst case (unconditional bound, does not depend upon assumptions like PLS $\neq$ P). As the set of best-response potential games is a super-set of exact potential games, these hardness results hold for best-response potential games as well.
For lottery contests with homogeneous agents, our fast convergence result for the randomized model implies polynomial-time convergence for best-case and average-case analysis, which is in contrast to the lack of such results for general best-response potential games. The worst-case rate of convergence for lottery contests is currently unknown.

\cite{moulin1978strategically} implicitly shows strategic equivalence between contests and zero-sum games. This directly implies convergence of fictitious play dynamics for two agents~\cite{ewerhart2020fictitious}, but no such result has been proven for three or more agents. A lottery contest corresponds to a Cournot game with isoelastic inverse demand and constant marginal costs. There are convergence results of learning dynamics for specific types of Cournot games, like Cournot oligopoly with strictly declining best-response functions~\cite{deschamps1975algorithm,thorlund1990iterative}, Cournot game with linear demand~\cite{slade1994does}, aggregative games that allow monotone best-response selections~\cite{huang2002fictitious,dubey2006strategic,jensen2010aggregative}, and others~\cite{dragone2012static,bourles2017altruism}. However, all these methods do not apply to the lottery contest whose best-response function is not monotone~\cite{dixit1987strategic}. 
A different line of research has studied the convergence (or chaotic behavior) of learning dynamics in other types of contests (like all-pay auctions) and Cournot games (e.g.,~\cite{puu1991chaos,warneryd2018chaotic,cheung2021learning}), but these techniques and results also do not apply to lottery contests.




%% file: 2prelim.tex
\section{Preliminaries}\label{sec:prelim} 
Let $[n] = \{1, 2, \ldots, n\}$. Let $\lg(x) = \log_2(x)$ and $\ln(x) = \log_e(x)$.

\subsection{Lottery Contest}
A lottery contest has $n$ agents and one prize of unit value (normalized). The agents simultaneously produce output $x = (x_1, \ldots, x_n) \in \bR_{\ge 0}^n$. 
Agent $i$ has a linear cost function and incurs a cost of $c_i x_i$ for producing output $x_i$, where $c_i > 0$. 
Agent $i$ gets an amount of prize proportional to $x_i$ if at least one agent produces a strictly positive output, else it is $1/n$.\footnote{Some papers in the literature, e.g., \cite{dasgupta1998designing}, assume that all agents get a prize of $0$ if they all produce a $0$ output. Our analysis and results are exactly the same for this alternate assumption as well.} The utility of agent $i$, denoted by $u_i(x)$, is
\begin{equation}\label{eq:utility:single}
    u_i(x) = \frac{x_i}{\sum_j x_j} - c_i x_i.
\end{equation}
An equivalent formulation of this model has $\tilde{u}_i(x) = \frac{V_i x_i}{\sum_j x_j} - x_i$, where $V_i = 1/c_i$ and where we scale up the utility function by the constant factor $V_i$, which does not affect the strategies of the agents. 

\subsection{Best-Response Dynamics}\label{sec:prelim:br}
Given the action profile $x_{-i} = (x_1, \ldots, x_{i-1}, x_{i+1}, \ldots, x_n)$ of all agents except $i$, the best response of agent $i$ is an action $x_i$ such that
\begin{equation*}
    x_i \in \argmax_{z \ge 0} u_i(z, x_{-i}) = \argmax_{z \ge 0} \frac{z}{z + \sum_{j \neq i} x_j} - c_i z.
\end{equation*}
Agent $i$ has a unique best response if $x_j > 0$ for some $j \neq i$, i.e., if the output produced by at least one other agent $j$ is non-zero. This unique best response can be computed by taking a derivative of $u_i(z, x_{-i})$ with respect to $z$, and is equal to $\sqrt{(\sum_{j \neq i} x_j) / c_i} - (\sum_{j \neq i} x_j)$ if $\sum_{j \neq i} x_j < 1/c_i$ and $0$ otherwise. On the other hand, an agent has no best response if the output produced by every other agent is $0$. If $x_{-i} = 0$, then by producing an output of $\epsilon > 0$, agent $i$ gets an utility of $u_i(\epsilon,0) = 1 - c_i \epsilon$, which is strictly more than $u_i(0,0) = 1/n$ for small $\epsilon$, so $x_i = 0$ cannot be a best response. Further, any $\epsilon > 0$ cannot be a best response because, $u_i(\epsilon/2,0) = 1 - c_i \epsilon/2 > 1 - c_i \epsilon = u_i(\epsilon,0)$.
To circumvent this technical issue, we assume that there is a very small positive value $a < 1/4$ such that $x_i = a$ is the best response of any agent $i$ if $x_{-i} = 0$. Formally, let $BR_i(x_{-i})$ denote the best response of agent $i$ given $x_{-i}$, which can be written as
\begin{equation}\label{eq:br:single}
    BR_i(x_{-i}) = 
    \begin{cases} 
        a, & x_{-i} = 0,  \\
        \sqrt{(\sum_{j \neq i} x_j) / c_i} - (\sum_{j \neq i} x_j), & 0 < \sum_{j \neq i} x_j \le 1/c_i, \\
        0, & \sum_{j \neq i} x_j > 1/c_i.
    \end{cases}
\end{equation}
We shall work with this assumption---an agent plays $a$ if everyone else plays $0$---throughout our paper.\footnote{This issue of not having any best response to $0$ can also be resolved by an alternate technical assumption: the prize is given to agent $i$ with probability $\frac{x_i}{b + \sum_j x_j}$ and agent $i$'s expected utility is $\frac{x_i}{b + \sum_j x_j} - c_i x_i$, where $b$ is a small positive constant. Notice that, with this assumption, the prize may not get allocated to any agent (or is allocated to a pseudo agent) with a positive probability of $\frac{b}{b + \sum_j x_j}$, unlike our model. We \textit{expect} all results in this paper to hold for this alternate model as well.\label{fn:alternateModel}}

Slightly overloading notation, let $x_t = (x_{t,i})_{i \in [n]}$ denote the action profile of the agents at time $t$ in the best-response dynamics. The best-response dynamics starts with an initial profile $x_0 = (x_{0,i})_{i \in [n]}$. At each time step $t \ge 0$, an agent $i_t \in [n]$ makes a best response move. 
Formally, $x_{t+1,i_t} = BR_{i_t}(x_{t, -i_t})$ and $x_{t+1, j} = x_{t, j}$ for $j \neq i_t$. In this paper, we study the convergence (or non-convergence) and the rate of convergence of this best-response dynamics.

\subsubsection{Randomized model for $n \ge 3$ agents.}\label{sec:prelim:random}
When there are just two agents, the best-response dynamics proceeds in a unique manner: the two agents make alternate best-response moves because making consecutive moves is redundant. But, when there are $n \ge 3$ agents, the best-response dynamics is not unique. At any time step, the $n-1 \ge 2$ agents who did not make the most recent best-response move can possibly make a non-trivial move. The rate of convergence depends upon which agent makes a move at a given time step.

We consider a randomized model for selecting the agent that makes the transition at any given time point $t$. In this random selection model, agent $i$ makes the best-response move with probability $w_{t,i}(x_t)$ at any time $t$. The probability $w_{t,i}(x_t)$ models the relative activity of agent $i$ at time $t$ given the current profile $x_t$. 
We assume that $w_{t,i}(x_t) \le U < 1/2$ for all agents and that $w_{t,i}(x_t) \ge L > 0$ for all agents except the agent who made the last transition. 
Our analysis will be worst-case over all $w_{t,i}$ given the parameters $0 < L \le U < 1/2$.
If $w_{t,i}(x_t) = 0$ for some $i$ for all $t$ and $x_t$, then $i$ never moves and we may never reach an equilibrium.\footnote{The best-response dynamics may converge with the remaining $n-1$ agents. We essentially have $n-1$ agents and a pseudo-agent with a constant output, similar to the discussion in footnote~\ref{fn:alternateModel}.}
The assumption $w_{t,i}(x_t) \le U < 1/2$ is required in one of the steps in our analysis (Lemma~\ref{lm:markov}), but is not unreasonable, especially for a large number of agents. For our current analysis, relaxing this assumption seems non-trivial, but may be achievable with an alternate analysis.

An important special case of our model is to assume that $w_{t,i}(x_t) = 1/n$ for all $t$, $i$, and $x_t$, i.e., the agent making the move is selected uniformly at random.

\subsection{Equilibrium}
A lottery contest always has a pure-strategy Nash equilibrium (which is also the unique equilibrium, including mixed-strategy Nash equilibria, see, e.g., \cite{vojnovic2015contest}, Chapter~4). So, we exclusively focus on pure equilibria in this paper.
\begin{definition}[Pure-Strategy Nash Equilibrium]
    An action profile $x = (x_1, \ldots, x_n)$ is a pure-strategy Nash equilibrium if it satisfies
    \[
        u_i (x_i, x_{-i}) \ge u_i (x_i', x_{-i}), 
    \]
    for every agent $i$ and every action $x_i'$ for agent $i$.
\end{definition}
In general, best-response dynamics in a lottery contest never exactly reaches the equilibrium, rather it may \textit{converge} to the equilibrium. The dynamics converges to an equilibrium if it reaches an $\epsilon$-approximate equilibrium in finite time for any $\epsilon > 0$.
\begin{definition}[Approximate Pure-Strategy Nash Equilibrium]
    An action profile $x = (x_1, \ldots, x_n)$ is an $\epsilon$-approximate pure-strategy Nash equilibrium, for $\epsilon > 0$, if it satisfies
    \[
        u_i (x_i, x_{-i}) \ge (1- \epsilon) u_i (x_i', x_{-i}), 
    \]
    for every agent $i$ and every action $x_i'$ for agent $i$.
\end{definition}

\subsection{Homogeneous Agents}
All agents have the same cost/utility function: for some $c > 0$, $u_i(x) = \frac{x_i}{\sum_j x_j} - c x_i$ for every $i \in [n]$. Notice that the best response of each agent is now identical given the same action profile of the remaining agents; let us, therefore, denote the best response as $BR = BR_i$ for every $i \in [n]$ by suppressing $i$ in \eqref{eq:br:single}.
There is a unique equilibrium where each agent $i$ plays an action (e.g., see \cite{vojnovic2015contest}, Chapter~4)
\begin{equation}\label{eq:equi:h:single}
    x^*_i = \frac{n-1}{n^2c}, \text{ for $i \in [n]$, at equilibrium.}
\end{equation}
For example, if there are only two agents, $n=2$, then $x^*_1 = x^*_2 = 1/(4c)$.

Notice that, w.l.o.g., we can assume that $c = 1$. This is achieved by a simple change of variable: we replace $x$ with $y = c x$ and $a$ with $c a$. Then the equilibrium is $y_i^* = c x_i^* = c\frac{n-1}{n^2c} = \frac{n-1}{n^2}$ for every $i$. The best response dynamics also changes similarly: 
\begin{itemize}
    \item \sloppy if $0 < \sum_{j \neq i_t} x_{t, j} \le 1/c \Longleftrightarrow 0 < \sum_{j \neq i_t} y_{t, j} = \sum_{j \neq i_t} c x_{t, j} \le 1$, then $y_{t+1, i_t} = c x_{t+1, i_t} = c \left(\sqrt{(\sum_{j \neq i} x_j) / c} - (\sum_{j \neq i} x_j) \right) = \sqrt{(\sum_{j \neq i} c x_j)} - (\sum_{j \neq i} c x_j) = \sqrt{(\sum_{j \neq i} y_j)} - (\sum_{j \neq i} y_j)$;
    \item if $\sum_{j \neq i_t} x_{t, j} > 1/c \Longleftrightarrow \sum_{j \neq i_t} y_{t, j} > 1$, then $y_{t+1, i_t} = c x_{t+1, i_t} = 0$;
    \item if $\sum_{j \neq i_t} x_{t, j} = 0  \Longleftrightarrow \sum_{j \neq i_t} y_{t, j} = 0$, then $y_{t+1, i_t} = c x_{t+1, i_t} = c a$.
\end{itemize}
So, $y_t$ behaves exactly like a Tullock contest where $c = 1$, and analyzing the game where $c = 1$ is equivalent to analyzing any game for arbitrary $c > 0$.

\subsection{A Potential Function for Homogeneous Agents}
In our analysis, in Section~\ref{sec:hom_n_ub}, we shall use the following potential function, which is similar to the one introduced by \citet{ewerhart2017lottery},
\begin{equation}\label{eq:potential}
    f(x) = \frac{1}{3} \left( \sum_{i} x_i \right)^3 - \sum_{i < j} x_i x_j + \frac{1}{6} \left(1 - \frac{1}{n} \right)^3,
\end{equation}
where $x = (x_i)_{i \in [n]} \in \bR_{\ge 0}^n$ is the action profile. Unlike usual notation, the action that \textit{minimizes} (and not maximizes) the potential function in equation \eqref{eq:potential} corresponds to the best response of an agent, given the action of other agents. We define $f$ this way to make our analysis a minimization problem and arguably easier to read. We shall prove relevant properties of $f$ in Section~\ref{sec:hom_n_ub}.

%% file: 3non-hom.tex
\section{Non-Homogeneous Agents}\label{sec:non-hom}

In this section, we provide examples with two non-homogeneous agents to demonstrate that best-response dynamics may not converge. These also implies non-convergence for three or more non-homogeneous agents. In the examples below, let the cost of agent $1$ be $1$ (normalized) and agent $2$ be $c \le 1$ (w.l.o.g.). Notice that agent $2$ has a lower cost and is stronger than agent $1$.

\begin{example}\label{ex:nonhom:1}
Let $c = 1/10$. Let $a = 10^{-5}$. The agents start from the initial profile $x_0 = (0, a)$, which leads to the cycle given in Table~\ref{tab:nonhom}.
\begin{table}[H]
    \centering
    \begin{tabular}{| c | c | c | c | c | c | c | c | c | c }
        \hline
        $x_{t,1}$ & $0.00000$ & $0.00315$ & $0.00315$ & $0.24321$ & $0.24321$ & $0.00000$ & $0.00000$ & $0.00315$ & $\ldots$ \\
        \hline
        $x_{t,2}$ & $0.00001$ & $0.00001$ & $0.17439$ & $0.17439$ & $1.31631$ & $1.31631$ & $0.00001$ & $0.00001$ & $\ldots$ \\
        \hline
    \end{tabular}
    \caption{Non-convergence of best-response dynamics for non-homogeneous agents (Example~\ref{ex:nonhom:1}).}
    \label{tab:nonhom}
\end{table}
\end{example}
The intuition for the cycle in Example~\ref{ex:nonhom:1} is that the stronger (lower cost $c = 1/10$) agent's best response may overshoot the upper limit on the output of the weaker agent (cost $1$), which makes the weaker agent best-respond with $0$. Then, the stronger agent plays $a$, then the weaker agent plays $\sqrt{a} - a$, and so on, and the output keeps on increasing until the stronger agent overshoots the upper limit of the weaker agent again. Such an issue never occurs if the agents are homogeneous.

We know that when $c = 1$, i.e., the two agents are homogeneous, we have convergence to the equilibrium. On the other hand, Example~\ref{ex:nonhom:1} gave a cycle for $c = 1/10$. So, a natural question is: how non-homogeneous should the agents be to get a cycle? Example~\ref{ex:nonhom:2} below shows that if $c \le 4/25 = 1/6.25$, then we can get a cycle. For $c \in (4/25,1)$, our simulations show convergence. We defer a formal analysis of the dynamics with almost homogeneous agents for future work.

\begin{example}\label{ex:nonhom:2}
Let $c = 4/25$. Notice that, if $a = 1/4$, we get the following cycle:
\[
    (0, 1/4) \longrightarrow (1/4, 1/4) \longrightarrow (1/4, 1) \longrightarrow (0, 1) \longrightarrow (0, 1/4) \longrightarrow \ldots.
\]
We next show that such cycles can also be constructed for certain arbitrarily small values of $a$. Notice that we can write a reversed best-response dynamics as follows: If agent $1$ makes the move at time $t$, then we have $x_{t+1,1} = \sqrt{x_{t,2}} - x_{t,2}$, which can be reversed as $x_{t,2} = \frac{(1 \pm \sqrt{1 - 4 x_{t+1,1}})^2}{4}$ (notice that there can be two possible values of $x_{t,2}$ (the $\pm$ sign) that lead to the same $x_{t+1,1}$; we will focus on the smaller $x_{t,2}$). Similarly, if agent $2$ made the move at time $t$, then we have $x_{t+1,2} = \sqrt{\frac{x_{t,1}}{c}} - x_{t,1}$, which can be written as $x_{t,1} = \frac{(1 \pm \sqrt{1 - 4 c x_{t+1,2}})^2}{4c}$. So, starting from $(1/4, 1/4)$ and going backward, we can have the following reversed sequence going to $0$:
\[
    \left(\frac{1}{4}, \frac{1}{4}\right) \longleftarrow  \left(\approx \frac{1.1}{100}, \frac{1}{4}\right) \longleftarrow  \left(\approx \frac{1.1}{100}, \approx \frac{1.2}{10^4}\right) \longleftarrow  \left(\approx \frac{2.3}{10^9}, \approx \frac{1.2}{10^4}\right) \longleftarrow  \left(\approx \frac{2.3}{10^9}, \approx \frac{5.5}{10^{18}}\right) \longleftarrow \ldots.
\]
By selecting one of the $x_{t,2}$ values as $a$ in the (reversed) sequence above, we can construct a cycle.
\end{example}
If $c < 4/25$, we can extend the examples above to have a set of values of $a$ that has a strictly positive measure in $[0,b]$, for any $b > 0$, and that leads to a cycle. See Example~\ref{ex:nonhom:3}. In other words, the non-convergence result is \textit{generic}. Note that the set of values of $a$ that cause a cycle can never cover the entire domain because there is always the equilibrium point $(c/(1+c)^2, 1/(1+c)^2)$ (and, likely, a neighborhood of this point of positive measure).
\begin{example}\label{ex:nonhom:3}
Let $c = 1/100$. Starting from $x_{t,2} \ge 1$, we track the reverse best-response dynamics (as done in Example~\ref{ex:nonhom:2}) as follows: $x_{t-1,1} \in [0.0098, 1)$, $x_{t-2,2} \in [9.80/10^5, 0.9803]$, $x_{t-3,1} \in [9.61/10^{11}, 0.0094]$, and so on. If we set $a$ equal to a point in the intervals corresponding to $x_{\tau,2}$, for some $\tau$, we will get a cycle.
\end{example}


%% file: 4hom-2.tex
\section{Two Homogeneous Agents}\label{sec:hom_2_agents}
In this section, we prove the rate of convergence of best-response dynamics for two homogeneous agents.
The initial state is $x_0 = (x_{0,1}, x_{0,2})$, and the state after $t \ge 0$ best-response moves is $x_t = (x_{t,1}, x_{t,2})$. 
We assume that agent $1$ makes the best-response move when $t$ is odd, i.e., $i_t = 1$ if $t$ odd, and agent $2$ makes the best-response move when $t$ is even, i.e., $i_t = 2$ if $t$ even. This is w.l.o.g. because: (i) it is redundant for an agent to play two consecutive best-response moves, therefore the two agents should alternate; (ii) as the agents are symmetric, it does not matter who makes the first move. Formally, for $t = 0, 1, 2, \ldots$, the best response dynamics goes as follows:
\begin{align*}
    x_{t+1, 2} &\in BR(x_{t, 1}), \text{ and } x_{t+1, 1} = x_{t, 1},\quad \text{ if $t = 0, 2, 4, \ldots$,} \\
    x_{t+1, 1} &\in BR(x_{t, 2}), \text{ and } x_{t+1, 2} = x_{t, 2},\quad \text{ if $t = 1, 3, 5, \ldots$.}
\end{align*}
where $BR = BR_1 = BR_2$ as defined in equation \eqref{eq:br:single}.

Next, we prove our main technical lemma for two agents that we use to prove the $O(\log\log)$ convergence of best-response dynamics.
\begin{lemma}\label{lm:2hom}
Let $z_0, z_1, z_2, \ldots$ be a sequence defined as follows:
\begin{align*}
    z_0 &= \gamma, \qquad\qquad\qquad\text{where $0 < \gamma < \frac{1}{2}$,} \\
    z_{t+1} &= \sqrt{z_t (1 - z_t)}, \quad\ \text{for $t \ge 0$. }
\end{align*}
Notice that $z_{t+1}$ is the geometric mean of $z_t$ and $(1-z_t)$.
The sequence $(z_t)_{t \ge 0}$ satisfies the following properties:
\begin{enumerate}
    \item $0 < z_{t} < z_{t+1} < \frac{1}{2}$ for $t \ge 0$,
    \item for any $0 < \epsilon < \frac{1}{2}$, $z_t \ge \frac{1}{2} - \epsilon$ if and only if $t \ge \tau$, where $\tau = \lg\lg(\frac{1}{\gamma}) + \lg\lg(\frac{1}{\epsilon}) + \Theta(1)$.
\end{enumerate}
\end{lemma}
\input{proofs/lm-2hom.tex}


\begin{theorem}\label{thm:hom2}
Best-response dynamics in lottery contests with two homogeneous agents reaches an $\epsilon$-approximate equilibrium in $\lg\lg(\frac{1}{\epsilon}) + \lg\lg(\frac{1}{\gamma}) + \Theta(1)$ steps, where $\gamma$ is a function of the initial state: $\gamma = \sqrt{x_{0,1}}$ if $0 < x_{0,1} < \frac{1}{4}$, $\gamma = x_{0,1} - \sqrt{x_{0,1}}$ if $\frac{1}{4} \le x_{0,1} < 1$, and $\gamma = \sqrt{a}$ otherwise.
\end{theorem}
Proof of Theorem~\ref{thm:hom2} is provided in the appendix.

%% file: proofs/lm-2hom.tex
\begin{proof} 
Let us prove the first property---the sequence is strictly increasing and is in $(0,1/2)$---by induction. Notice that $z_0 = \gamma \in (0,\frac{1}{2})$ by definition. Let $z_t = \frac{1-\zeta}{2}$ for some $0 < \zeta < 1$. Now, $z_{t+1} = \sqrt{(\frac{1-\zeta}{2})(1 - \frac{1-\zeta}{2})} = \sqrt{(\frac{1-\zeta}{2})(\frac{1+\zeta}{2})} = \frac{\sqrt{1-\zeta^2}}{2}$, therefore $0 < z_{t+1} < \frac{1}{2}$. Also,
\begin{equation*}
    0 < \zeta < 1 \implies 2\zeta^2 < 2 \zeta \implies 1 + \zeta^2 - 2 \zeta < 1 - \zeta^2 \implies 1 - \zeta < \sqrt{1 - \zeta^2} \implies z_t < z_{t+1}.
\end{equation*}

We proved above that the sequence $(z_t)$ is increasing. We want to find the number of steps for the sequence to increase from $\gamma$ to $\frac{1}{2} - \epsilon$. We break this into two parts: the number of steps required to reach (i) $\frac{1}{4}$ from $\gamma$ in Lemma~\ref{lm:2hom:1}; (ii) $\frac{1}{2} - \epsilon$ from $\frac{1}{4}$ in Lemma~\ref{lm:2hom:2}.

\begin{lemma}\label{lm:2hom:1}
    Given $z_0 = \gamma < \frac{1}{4}$, $z_t \ge \frac{1}{4}$ for all $t \ge \lg\lg(\frac{1}{\gamma})$ and $z_t < \frac{1}{4}$ for all $t < \lg\lg(\frac{1}{\gamma}) - 1$.
\end{lemma}
\begin{proof}
    We know that $z_{t+1} = \sqrt{z_t (1 - z_t)}$. Let us first prove the upper bound on the time it takes to reach $\frac{1}{4}$. As $z_t \le \frac{1}{2}$ for every $t$, therefore $1 - z_t \ge \frac{1}{2}$, which implies
    \begin{align*}
        z_{t} = \sqrt{z_{t-1} (1-z_{t-1})} \ge \left( \frac{1}{2} \right)^{\frac{1}{2}} z_{t-1}^{\frac{1}{2}} \ge \left( \frac{1}{2} \right)^{\frac{1}{2} + \frac{1}{4}} z_{t-2}^{\frac{1}{4}} \ge 
        \ldots \ge \left( \frac{1}{2} \right)^{\sum_{j = 1}^{t} \frac{1}{2^j}} z_0^{\frac{1}{2^t}} \ge \frac{\gamma^{\frac{1}{2^t}}}{2}.
    \end{align*}
    Now, we want $z_t \ge \frac{1}{4}$. So, if we ensure that $\frac{\gamma^{\frac{1}{2^t}}}{2} \ge \frac{1}{4}$, then we are good.
    \begin{equation*}
        \frac{\gamma^{\frac{1}{2^t}}}{2} \ge \frac{1}{4} \Longleftrightarrow \lg(\gamma) \frac{1}{2^t} - 1 \ge -2 \Longleftrightarrow \lg(\frac{1}{\gamma}) \frac{1}{2^t} \le 1 \Longleftrightarrow t \ge \lg\lg(\frac{1}{\gamma}).
    \end{equation*}

    Let us now focus on the lower bound. As $z_t \ge 0$ for every $t$,
    \begin{equation*}
        z_{t} = \sqrt{z_{t-1} (1-z_{t-1})} \le \sqrt{z_{t-1}} \le z_{t-2}^{\frac{1}{4}} \le \ldots \le z_0^{\frac{1}{2^t}} = \gamma^{\frac{1}{2^t}}. 
    \end{equation*}
    Now, we want $z_t < \frac{1}{4}$. So, if we ensure that $\gamma^{\frac{1}{2^t}} < \frac{1}{4}$, then we are done.
    \begin{equation*}
        \gamma^{\frac{1}{2^t}} < \frac{1}{4} \Longleftrightarrow \lg(\gamma) \frac{1}{2^t} < -2 \Longleftrightarrow \lg(\frac{1}{\gamma}) \frac{1}{2^t} > 2 \Longleftrightarrow t < \lg\lg(\frac{1}{\gamma}) - 1.
    \end{equation*}
\end{proof}
In the next lemma, we bound the number of steps required to reach $\frac{1}{2} - \epsilon$, starting from $\frac{1}{4}$.
\begin{lemma}\label{lm:2hom:2}
    If $z_0 \ge \frac{1}{4}$, then $z_t \ge \frac{1}{2} - \epsilon$ for all $t \ge \lg\lg(\frac{1}{\epsilon})$. On the other hand, if $z_0 \le \frac{1}{4}$, then $z_t < \frac{1}{2} - \epsilon$ for all $t < \lg\lg(\frac{1}{\epsilon}) - 2$.
\end{lemma}
\begin{proof}
    Let us define $\zeta_t = 1 - 2 z_t \Longleftrightarrow z_t = \frac{1-\zeta_t}{2}$. Notice that $0 < \zeta_t < 1$ as $0 < z_t < \frac{1}{2}$. Also, as $z_t$ strictly increases, $\zeta_t$ strictly decreases.

    Let us prove the upper bound on the time it takes to reach $\frac{1}{2} - \epsilon$.
    We want to find the time $t$ such that $\frac{1-\zeta_t}{2} = z_t \ge \frac{1}{2} - \epsilon \Longleftrightarrow \zeta_t \le 2 \epsilon$.
    As $z_{t+1} = \sqrt{z_t (1-z_t)}$, we have 
    \[
        \frac{1-\zeta_{t+1}}{2} = \sqrt{\frac{1-\zeta_t}{2} \left(1 -  \frac{1-\zeta_t}{2} \right)} \Longleftrightarrow 1 - \zeta_{t+1} = \sqrt{1 - \zeta_t^2}.  
    \]
    Now, as $0 < 1 - \zeta_t^2 < 1$ for all $t$, we have $\sqrt{1 - \zeta_t^2} > 1 - \zeta_t^2$ for all $t$, which implies, 
    \[
        1 - \zeta_{t} > 1 - \zeta_{t-1}^2 \implies \zeta_{t} < \zeta_{t-1}^2 < \zeta_{t-2}^4 < \ldots < \zeta_0^{2^t}.
    \]
    If $z_0 \ge \frac{1}{4}$, then $\zeta_0 \le \frac{1}{2}$, so $\zeta_{t} < \zeta_0^{2^t} < \frac{1}{2^{2^t}}$. Now, we want $z_t \ge \frac{1}{2} - \epsilon$, which equivalent to $\zeta_t \le 2 \epsilon$, which is implied by
    \[
        \frac{1}{2^{2^t}} \le 2 \epsilon \Longleftrightarrow -2^t \le 1 + \lg(\epsilon) \Longleftrightarrow 2^{t} \ge \lg(\frac{1}{\epsilon}) - 1 \impliedby t \ge \lg\lg(\frac{1}{\epsilon}).
    \]

    Let us similarly find the lower bound on time taken to reach $\frac{1}{2} - \epsilon$. We want $z_t < \frac{1}{2} - \epsilon \Longleftrightarrow \zeta_t > 2 \epsilon$. As $1 + x \le e^x$ for all $x \in \bR$ and $e^{-x} \le 1 - \frac{x}{2}$ for $x \in [0, 1]$, we have
    \begin{multline*}
        1 - \zeta_{t+1} = \sqrt{1 - \zeta_{t}^2} \le \sqrt{e^{-\zeta_t^2}} = e^{-\frac{\zeta_t^2}{2}} \le 1 - \frac{\zeta_t^2}{4} \implies \zeta_{t+1} \ge \frac{\zeta_t^2}{4} \\
        \implies \zeta_t \ge \frac{\zeta_{t-1}^2}{4} \ge \left( \frac{1}{4} \right)^{1 + 2} \zeta_{t-2}^4 \ge \left( \frac{1}{4} \right)^{1 + 2 + 4} \zeta_{t-3}^8 \ge \ldots \ge \left( \frac{1}{4} \right)^{\sum_{j = 1}^{t} 2^{j-1}} \zeta_0^{2^t} \ge  \left( \frac{\zeta_0}{4} \right)^{2^t}.
    \end{multline*}
    We want $\zeta_t > 2 \epsilon$, which is ensured by
    \begin{multline*}
        \left( \frac{\zeta_0}{4} \right)^{2^t} > 2 \epsilon \Longleftrightarrow 2^t (\lg(\zeta_0) - 2) > 1 + \lg(\epsilon) \impliedby 2^t < \frac{\lg(\frac{1}{\epsilon}) - 1}{2 + \lg(\frac{1}{\zeta_0})} \\
        \impliedby 2^t < \frac{\lg(\frac{1}{\epsilon}) - 1}{2 + 1} \impliedby 2^t < \frac{\lg(\frac{1}{\epsilon})}{4}
        \impliedby t < \lg\lg(\frac{1}{\epsilon}) - 2,
    \end{multline*}
    where we used $\zeta_0 \ge \frac{1}{2}$ to derive $\lg(\frac{1}{\zeta_0}) \le 1$ and used $\epsilon \le \frac{1}{4}$ to derive $\lg(\frac{1}{\epsilon}) - 1 \ge \frac{4}{3}\lg(\frac{1}{\epsilon})$.
\end{proof}
Combining Lemmas~\ref{lm:2hom:1}~and~\ref{lm:2hom:2}, we know that in exactly
$\lg\lg(\frac{1}{\gamma}) + \lg\lg(\frac{1}{\epsilon}) + \tau$ steps, where $\tau \in \{-3, -2, -1, 0\}$, the sequence $z_t$ reaches $\ge \frac{1}{2} - \epsilon$ starting from $\gamma$.
\end{proof}

%% file: 5hom-n-ub.tex
\section{Three or More Homogeneous Agents}\label{sec:hom_n_ub}
We now study the case when there are $n \ge 3$ agents. As introduced in Section~\ref{sec:prelim:random}, we consider a randomized model for selecting the agent who makes the transition at any given time point $t$. Agent $i$ makes the transition at time $t$ w.p. (with probability) $w_{t,i}(x_t)$ given action profile at $t$ is $x_t$. We assume $0 < L \le w_{t,i}(x_t) \le U < 1/2$ (the lower bound is not necessary for the agent who played at time $t-1$). An important special case of our model is to assume that $w_{t,i}(x_t) = 1/n$ for all $t$, $i$, and $x_t$, i.e., the agent making the move is selected uniformly at random. We do a worst-case analysis over all $w_{t,i}$ given the parameters $0 < L \le U < 1/2$.




\begin{theorem}\label{thm:hom}
\sloppy
Best-response dynamics in lottery contests with $n \ge 3$ homogeneous agents and randomized selection of agents reaches an $\epsilon$-approximate equilibrium with probability $1-\delta$ in 
$O\left( \frac{1}{1 - 2U} \log\log\left( \frac{1}{\gamma} \right) 
    + \frac{1}{L (1-2U)} \log\left( \frac{n}{\epsilon \delta}\right) 
    + \frac{1}{(1-2U)^2} \log\left( \frac{1}{\delta} \right) \right) $ 
steps for every $\epsilon, \delta \in (0,1)$, where $\gamma$ is a function of the initial state $\gamma = \min( A \cup B \cup \{a\} )$, where $A = \{ x_{0,i} \mid 0 < x_{0,i} < 1, i \in [n] \}$ and $B = \{ \sqrt{(1-\sqrt{y})/2} \mid y \in A \}$. 
As a lower bound, we show that the convergence takes at least 
$\Omega\left( \left(n + \frac{1}{L}\left(1 - \frac{1}{nU}\right) \right) \log(n \delta) + \log\log\left(\frac{1}{\epsilon}\right) + \log\log\left(\frac{1}{\gamma}\right) \right)$ 
steps w.p. $1 - \delta$.
\end{theorem}
Notice that the above theorem implies, when agents are selected uniformly at random at each time step (i.e., $L = U = 1/n$), the following corollary.
\begin{corollary}\label{thm:hom:unif}
Best-response dynamics in lottery contests with $n \ge 3$ homogeneous agents and uniform selection of agents reaches an $\epsilon$-approximate equilibrium w.p. $1-\delta$ in 
$O\left( \log\log\left( \frac{1}{\gamma} \right) + n \log\left( \frac{n}{\epsilon \delta}\right) \right)$ 
steps for every $\epsilon, \delta \in (0,1)$, where $\gamma$ is as defined in Theorem~\ref{thm:hom}. 
Also, the convergence takes at least 
$\Omega\left( n \log(n \delta) + \log\log\left(\frac{1}{\epsilon}\right) + \log\log\left(\frac{1}{\gamma}\right) \right)$ 
steps w.p. $1 - \delta$.
\end{corollary}

\begin{proof}[Proof of Theorem~\ref{thm:hom}]
We shall use an extension of the well-known coupon collector problem in the first part of our analysis. In the coupon collector problem, we select a coupon out of $n$ coupons uniformly at random at each time step, and we want to bound the time it takes to collect all coupons. A tight $\Theta(n \log n)$ high probability bound is known for this problem (see, e.g., \cite{mitzenmacher2017probability}). We use this result to derive the following bounds on the time it takes for each agent to play at least once in the best-response dynamics; the proof is in the appendix.

\begin{lemma}\label{lm:coupon}
    Let $T$ be the time it takes for all agents to play at least once in the best-response dynamics, we have the following high probability bounds: (i) upper bound, $\Prob[T \le \frac{1}{L}\ln (\frac{n}{\delta})] \ge 1 - \delta$; (ii) lower bound, $\Prob[T \ge \max(\frac{1}{L}(1 - \frac{1}{nU}), n) \log(n \delta)] \ge 1 - \delta$.
\end{lemma}

Observe that, to reach an approximate equilibrium, every agent must make at least one best-response move. For example, it is easy to check that there can be no approximate equilibrium with any agent producing an output of $1$: if an agent is playing $1$, then everyone else must play $0$; if everyone else is playing $0$, then this agent would want to deviate and not play $1$. So, if every agent starts with an output of $1$, then each agent must make at least one move to reach an approximate equilibrium. Using Lemma~\ref{lm:coupon}, we get the lower bound of $\Omega( (n + \frac{1}{L}(1 - \frac{1}{nU})) \log(n \delta) )$. Moreover, from Theorem~\ref{thm:hom2}, we have a $\Omega(\log(\log(\frac{1}{\epsilon})\log(\frac{1}{\gamma})))$ lower bound for two agents, which automatically applies to $n \ge 3$ agents. 


The rest of the proof is dedicated to proving a high probability upper bound.
Keeping the same notation as our previous discussions, let $x_{t,i}$ be the output of agent $i \in [n]$ at time $t \ge 0$. Note that $x_{t,i}$ is a random variable and $(x_t)_{t \ge 0}$ is a stochastic process because the agent $i_t$ making the move at time $t$ is selected randomly. Let $s_t = \sum_{i} x_{t,i}$ be the total output produced by the agents at time $t$; $s_t$ is also stochastic.

We call an initial time period of the best response dynamics the \textit{warm-up} phase (Definition~\ref{def:warmup}).
When the warm-up phase ends, we ensure that the output profile is, and stays thereafter, in a certain well-behaved region.
\begin{definition}[Warm-Up Phase]\label{def:warmup}
    The time period $[0,T_{warm}) = \{0, 1, \ldots, T_{warm}-1\}$ denotes the warm-up phase, where $T_{warm}$ is the smallest time such that for every $t \ge T_{warm}$:
    \begin{enumerate}
        \item \label{def:warmup:1} all agents produce output of at most $1/4$: $0 \le x_{t,i} \le 1/4$ for every $i$;
        \item \label{def:warmup:2} there are at least two agents $i$ and $j \neq i$ with positive output: $x_{t,i} > 0$ and $x_{t,j} > 0$;
        \item \label{def:warmup:3} the total output is strictly less than $1$: $s_t < 1$.
    \end{enumerate}
\end{definition}

There are three main parts of the analysis. In the first part, we bound the time it takes for the warm-up phase to finish with high probability. The next two parts assume the completion of the warm-up phase. In the second part, we bound the time it takes for the total output $s_t$ to be at least $1/4$ for a sufficiently large number of time steps. We do it by discretizing the space of $s_t$ and analyzing the Markov chain associated with it. Finally, in the third part of the analysis, we show that, if $s_t \ge 1/4$ for a sufficiently large number of steps, then the output profile is close to the equilibrium output profile with high probability. Here, we use the properties of the potential function $f$. All omitted proofs are in the appendix.

\paragraph{\textbf{Analysis Part 1 (Warm-Up Phase).}}
We shall use the following properties (Lemma~\ref{lm:prop1}) of best-response dynamics in our subsequent analysis.
\begin{lemma}\label{lm:prop1}
    The best-response dynamics satisfies the following properties:
    \begin{enumerate}
        \item \label{lm:prop1:1} $s_t > 0$ for every $t \ge 1$.
        \item \label{lm:prop1:2} $s_t < 1 \implies s_{t+1} < 1$ for every $t \ge 0$.
        \item \label{lm:prop1:3} For $i,j \in [n]$ such that $i \neq j$, if $x_{t,i} > 0$, $x_{t,j} > 0$, and $s_t < 1$, then $x_{t+1,i} > 0$ and $x_{t+1,j} > 0$ for every $t \ge 0$.
    \end{enumerate}
\end{lemma}

The next lemma gives a high probability bound on the time taken by the warm-up phase.
\begin{lemma}\label{lm:part1}
    Time till completion of the warm-up phase $T_{warm} = O(\frac{1}{L}\log(\frac{n}{\delta}))$ w.p. $1 - \delta$.
\end{lemma}

In Lemma~\ref{lm:part1}, we proved that the warm-up phase finishes with high probability in about $O(\frac{1}{L}\log n)$ steps. In the next lemma (Lemma~\ref{lm:outputlb}), we lower bound the total output when the warm-up phase finishes, i.e., at time $T_{warm}$, as a function of the initial profile $x_0$. This bound will be used in our subsequent analysis.

\begin{lemma}\label{lm:outputlb}
$s_{T_{warm}} \ge \gamma = \min( A \cup B \cup \{a\} )$, where $A = \{ x_{0,i} \mid 0 < x_{0,i} < 1, i \in [n] \}$ and $B = \{ \sqrt{(1-\sqrt{y})/2} \mid y \in A \}$. 
\end{lemma}

This completes the first part of our analysis. In the second part, we bound the time it takes to have $s_t$ above $1/4$ for a large number of time steps with high probability, given we already have completed the warm-up phase. We do it using the Markov chain below.

\paragraph{\textbf{A biased random walk with a wall.}} Let us consider a Markov chain with countably infinite states denoted by positive integers $\bZ_{> 0}$. Let $y_t \in \bZ_{> 0}$ be a random variable that denotes the state of this Markov chain at time $t \ge 0$. The state moves left (or stays at $1$ if it is already at $1$) w.p. $(1-p)$ and moves right w.p. $p$, where $p < 1/2$. Formally, the transition function of the chain is: $\Prob[y_{t+1} = 1 \mid y_{t} = 1] = 1 - p$, $\Prob[y_{t+1} = 2 \mid y_{t} = 1] = p$, $\Prob[y_{t+1} = i - 1 \mid y_{t} = i] = 1 - p$, and $\Prob[y_{t+1} = i+1 \mid y_{t} = i] = p$ for $i \in \bZ_{\ge 2}$. Figure~\ref{fig:markov} shows the Markov chain.
\begin{figure}[htbp]
\centering
\resizebox{!}{64pt}{\input{fig-markov.tex}}
\caption{Markov chain}
\label{fig:markov}
\end{figure}
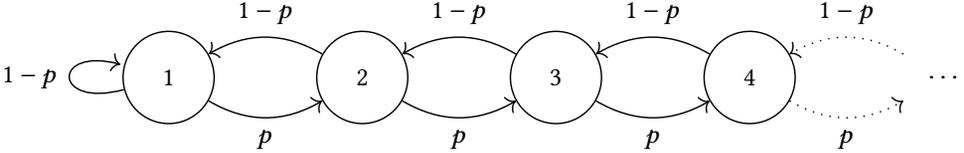
We bound the time it takes for the stochastic process $y_t$ to visit state $1$ sufficiently large number of times with high probability, which we do using suitable concentration bounds in Lemma~\ref{lm:markov}.
\begin{lemma}\label{lm:markov}
    Starting from initial state $y_0 = k$, the Markov chain visits state $1$ at least $m$ times w.p. at least $1-\delta$ after $\frac{4}{1-2p} \max\left( m + k, \frac{1}{1-2p} \ln\left( \frac{1}{\delta} \right) \right)$ time steps.
\end{lemma}

\paragraph{\textbf{Analysis Part 2 (Total Output).}}
We now use the Markov chain defined earlier to bound the time it takes after the warm-up phase for $s_t$ to be above $\frac{1}{4}$ for a large number of time steps with high probability.

\begin{lemma}\label{lm:part2}
Within $T = T_{warm} + O\left( \frac{1}{1-2U}m + \frac{1}{1-2U} \log\log\left( \frac{1}{\gamma} \right) + \frac{1}{(1-2U)^2} \log\left( \frac{1}{\delta} \right) \right)$ time steps, w.p. $1-\delta$,
there must have been at least $m$ times steps where $s_t \ge \frac{1}{4}$,
where $\gamma$ is the lower bound on $s_{T_{warm}}$ as stated in Lemma~\ref{lm:outputlb}; formally, $\left| \left\{ t \mid T_{warm} \le t \le T, s_t \ge \frac{1}{4} \right\} \right| \ge m$ w.p. $1-\delta$.
\end{lemma}
\input{proofs/part2.tex}

\paragraph{\textbf{Analysis Part 3 (Potential Function).}}
In the third part of the analysis, we shall make use of the potential function $f$ given in equation \eqref{eq:potential}. Let us first prove some properties of the potential function $f$. Let $z = (z_i)_{i \in [n]} \in \bR_{\ge 0}^n$ such that $\sum_{i \in [n]} z_i < 1$. Let $z^* = (z^*_i)_{i \in [n]}$ where $z^*_i = \frac{n-1}{n^2}$. Notice that $z^*$ is the equilibrium action profile. Let $\sigma = \sum_i z_i$ and $\sigma^* = \sum_i z_i^*$.
\begin{lemma}\label{lm:potential}
The potential function $f$ satisfies:
\begin{enumerate}
    \item \label{lm:potential:1} $f(z^*) = 0$;
    \item \label{lm:potential:2} $0 \le f(z) \le 1/2$;
    \item \label{lm:potential:3} $f(z) \ge \frac{1}{40}\left(\frac{n-1}{n}\right)^3$ if $\sum_{i} z_i \le \frac{3(n-1)}{4n}$;
    \item \label{lm:potential:4} $f(z) \ge \frac{1}{2}|| z - z^* ||_2^2 = \frac{1}{2} \sum_i (z_i - z^*_i)^2$ if $\sum_{i} z_i \ge \frac{3(n-1)}{4n}$;
    \item \label{lm:potential:5} $f(z) \le \frac{1}{2} \sum_i ( \sum_{j \neq i} ( z_i - z_i^* ) )^2 $.
\end{enumerate}
\end{lemma}

In Lemma~\ref{lm:potential} \eqref{lm:potential:4}, we showed that the $\ell_2$-distance of the output profile is close to the equilibrium profile if the potential function is close to $0$. In the Lemma~\ref{lm:lipschitz}, we show that we are at an approximate equilibrium if the output profile is close to the equilibrium profile.
\begin{lemma}\label{lm:lipschitz}
If $|| z - z^* ||_2 \le \epsilon$ for $0 \le \epsilon < \frac{1}{n\sqrt{n}}$, then $z$ is an $(3 \sqrt{n} \epsilon)$-approximate equilibrium.
\end{lemma}

Let us now go back to the best-response dynamics. In the next lemma, we prove that the potential converges to $0$ quickly with high probability.
\begin{lemma}\label{lm:part3}
$\Prob[f(x_t) \le \epsilon] \ge 1 - \delta$ after there have been at least $O\left(\frac{1}{ L (1-2U)} \log\left( \frac{1}{\epsilon \delta}\right)\right)$ time steps (after the warm-up phase) where total output $s_t \ge \frac{1}{4}$.
\end{lemma}
\input{proofs/part3.tex}

For $\epsilon < \frac{3}{n}$ and $s_t \ge \frac{3(n-1)}{n}$, from Lemma~\ref{lm:lipschitz}, we can ensure that $x_t$ is an $\epsilon$-approximate equilibrium by ensuring $|| x_t - z^* || \le \frac{1}{3 \sqrt{n}} \epsilon $, where $z^*$ is the equilibrium profile; which we can further guarantee, using Lemma~\ref{lm:potential} \eqref{lm:potential:4}, by ensuring that $f(x_t) \le \frac{1}{6 \sqrt{n}} \epsilon$. On the other hand, for $\epsilon \ge 3/n$ or $s_t \le \frac{3(n-1)}{n}$, we can simply plug in $\frac{3}{n}$ instead of $\epsilon$ and get an $(\frac{3}{n})$-approximate equilibrium, which is also an $\epsilon$-approximate equilibrium, by ensuring that $f(x_t) \le \frac{2}{n \sqrt{n}}$.


Let us now combine all our lemmas and finish the proof. We set the probabilities of failure in each of the Lemmas~\ref{lm:part1},~\ref{lm:part2},~and~\ref{lm:part3} to be $\delta/3$, ensuring total probability of failure is at most $\delta$ by union bound. So, we get an $\epsilon$-approximate equilibrium w.p. $1-\delta$ in steps:
\[
    O\left( \frac{1}{1 - 2U} \log\log\left( \frac{1}{\gamma} \right) 
    + \frac{1}{L (1-2U)} \log\left( \frac{n}{\epsilon \delta}\right) 
    + \frac{1}{(1-2U)^2} \log\left( \frac{1}{\delta} \right) \right) .
\]

\end{proof}

%% file: fig-markov.tex
\begin{tikzpicture}[
    start chain=main going right,
    state/.style={circle,minimum size=9mm,draw},
      node distance=10mm,
      font=\scriptsize,
      auto
    ]
    
    \node [state, on chain, fill=white, text=black] (A) {$1$};
    \node[state, on chain, fill=white, text=black]  (B) {$2$};
    \node[state, on chain, fill=white, text=black]  (C) {$3$};
    \node[state, on chain, fill=white, text=black]  (D) {$4$};
    \node[on chain, circle, minimum size=9mm, fill=white, text=black]  (E) {$\ldots$};

    \path[->, draw]
    (A) edge[loop left] node{$1-p$} (A);
    
    \path[->,draw, bend right=30]
    (A) edge node[below] {$p$} (B);
    
    \path[->,draw, bend right=30]
    (B) edge node[above] {$1 - p$} (A);

    \path[->,draw, bend right=30]
    (B) edge node[below] {$p$} (C);
    
     \path[->,draw, bend right=30]
    (C) edge node[above] {$1-p$} (B);

    \path[->,draw, bend right=30]
    (C) edge node[below] {$p$} (D);
    
     \path[->,draw, bend right=30]
    (D) edge node[above] {$1-p$} (C);

    \path[->,draw, dotted, bend right=30]
    (D) edge node[below] {$p$} (E);
    
     \path[->,draw, dotted, bend right=30]
    (E) edge node[above] {$1-p$} (D);

\end{tikzpicture}

%% file: proofs/part2.tex
\begin{proof} 
In this proof, we assume that the warm-up phase has completed and all conditions valid after the completion of the warm-up phase are satisfied. To make the analysis less cluttered, let us assume that $T_{warm} = 0$, i.e., $x_t$ satisfies all conditions for completion of the warm-up phase for all $t \ge 0$ (rather than $t \ge T_{warm}$). 


We map the evolution of $s_t$ to a discrete stochastic process $z_t$ that takes values in $\bZ_{>0}$:
\begin{itemize}
    \item We partition the unit interval $(0,1)$, range of $s_t$, into the following intervals (written in decreasing order, starting closer to $1$ and going to $0$):
    \[
        \left[\frac{1}{4}, 1\right), \left[\frac{1}{8}, \frac{1}{4} \right), \left[\frac{1}{32}, \frac{1}{8} \right), \ldots,  \left[\left(\frac{1}{2}\right)^{2^{\ell-1}+1}, \left(\frac{1}{2}\right)^{2^{\ell-2}+1} \right), \ldots,
    \]
    where the $\ell$-th interval for $\ell \ge 2$ is equal to $[(1/2)^{2^{\ell-1}+1}, (1/2)^{2^{\ell-2}+1} )$. Notice that the values in the intervals are decreasing at a double exponential rate.
    \item The state of the stochastic process $z_t$ is equal to $\ell$ if $s_t$ lies in the $\ell$-th interval. Formally,
    \[
        z_t = 
        \begin{cases}
            1, &\text{if $s_t \in [1/4, 1)$},\\
            \ell, &\text{if $s_t \in \left[(1/2)^{2^{\ell-1}+1}, (1/2)^{2^{\ell-2}+1} \right)$ where $\ell \ge 2$}.
        \end{cases}
    \]
\end{itemize}
We shall prove that the $z_t$ has a higher tendency to move left than $y_t$, where $y_t$ follows the Markov chain defined earlier (Figure~\ref{fig:markov}), given we set $p = U$ for $y_t$.
Remember, $y_t$ moves one step right ($y_{t+1} = y_t + 1$) w.p. $p = U$ and one step left ($y_{t+1} = \max(1, y_t - 1)$) w.p. $1-p = 1-U$. We shall show that $z_t$ moves at least one step left (or stays at $1$) w.p. at least $1-p$ and moves at most one step right w.p. at most $p$.
If we set $z_0 = y_0$, then this result ensures that the probability that $z_t$ visits state $1$ at least $m$ times is at least the probability that $y_t$ visits state $1$ at least $m$ times.


\textit{Left transitions.} We know that $s_{t+1} = \sqrt{s_t - x_{t,i}}$, where $i$ is selected from $[n]$ randomly. As $\sum_i x_{t,i} = s_t$, there can be at most one agent $j$ such that $x_{t,j} > s_t/2$. This agent $j$ is selected w.p. at most $U$. So, w.p. at least $1-U$ we select an agent $i$ such that $x_{t,i} \le s_t/2 \implies s_t - x_{t,i} \ge s_t/2 \implies s_{t+1} = \sqrt{s_t - x_{t,i}} \ge \sqrt{s_t/2}$.
\begin{itemize}
    \item If $z_t = 1$, then $s_t \ge 1/4$. So, $s_{t+1} \ge \sqrt{s_t/2} \ge \sqrt{1/8} \ge 1/4$ w.p. at least $1-U = 1-p$. So, $z_{t+1} = 1$ w.p. at least $1-p$.
    
    \item If $z_t = \ell$ for some $\ell \ge 2$, then $s_t \ge (1/2)^{2^{\ell-1}+1}$. So, $s_{t+1} \ge \sqrt{s_t/2} \ge \sqrt{(1/2)^{2^{\ell-1}+1}/2} = \sqrt{(1/2)^{2^{\ell-1}+2}} = (1/2)^{2^{\ell-2}+1}$ w.p. at least $1-p$. So, $z_{t+1} \le \ell-1$ w.p. at least $1-p$.
\end{itemize}

\textit{Right transitions.} We showed above that $z_t$ makes a left transition w.p. at least $1-p$. So, $z_t$ stays at its position or makes a right transition w.p. at most $p$, which is what we want. But we still need to show that $z_t$ can move at most one step right, i.e., $z_{t+1} \le z_t + 1$ w.p. $1$.

Let us focus on the transition at time $t$: $s_t \rightarrow s_{t+1}$, $z_t \rightarrow z_{t+1}$. Say agent $i = i_t$ makes this transition. 
If $i$ made the best-response move at time $t-1$ as well, i.e., $i_{t-1} = i$, then $s_{t+1} = s_t$ because consecutive moves are redundant. Therefore, $z_{t+1} = z_t \le z_t + 1$, as required.

Let us assume that $i_{t-1} \neq i$. Let $j = i_{t-1} \neq i$ be the agent that made the move at time $t-1$. As $i_{t-1} \neq i$, therefore $x_{t-1,i} = x_{t,i}$. Similarly, as $i_t \neq j$, therefore $x_{t,j} = x_{t+1,j}$. On the other hand, for $k \neq i,j$, notice that $x_{t-1, k} = x_{t, k} = x_{t+1, k}$. 

Let $\alpha = \sum_{k \neq i, j} x_{t,k}$. We can write $s_t - x_{t,i}$ using $\alpha$ as
\begin{equation}\label{eq:nagents:unif:1}
    s_t - x_{t,i} = x_{t,j} + \sum_{k \neq i,j} x_{t,k} = x_{t,j} + \alpha.
\end{equation}
We can also write $x_{t,j}$ using $\alpha$ as
\begin{align}\label{eq:nagents:unif:2}
    x_{t,j} &= \sqrt{\sum_{k \neq j} x_{t-1,k}} - \sum_{k \neq j} x_{t-1,k} = \sqrt{x_{t-1,i} + \alpha} - (x_{t-1,i} + \alpha) = \sqrt{x_{t,i} + \alpha} - (x_{t,i} + \alpha),
\end{align}
and $s_t$ using $\alpha$ as
\begin{equation}\label{eq:nagents:unif:3}
    s_t = \sqrt{s_{t-1} - x_{t-1,j}} = \sqrt{\sum_{k \neq j} x_{t-1,k}} = \sqrt{x_{t-1,i} + \alpha} = \sqrt{x_{t,i} + \alpha}.
\end{equation}
Combining \eqref{eq:nagents:unif:1} and \eqref{eq:nagents:unif:2}, we get
\begin{equation}\label{eq:nagents:unif:4}
    s_t - x_{t,i} = x_{t,j} + \alpha = \sqrt{x_{t,i} + \alpha} - (x_{t,i} + \alpha) + \alpha = \sqrt{x_{t,i} + \alpha} - x_{t,i}.
\end{equation}
We shall now lower bound the ratio between $s_{t+1}$ and $s_t$. From \eqref{eq:nagents:unif:3} and \eqref{eq:nagents:unif:4}, we have
\begin{equation*}
    \frac{s_{t+1}}{s_t} = \frac{\sqrt{s_{t} - x_{t,i}}}{s_t} = \frac{\sqrt{\sqrt{x_{t,i} + \alpha} - x_{t,i}}}{\sqrt{x_{t,i} + \alpha}} = \sqrt{\frac{1}{\sqrt{x_{t,i} + \alpha}} - \frac{x_{t,i}}{x_{t,i} + \alpha}}.
\end{equation*}
Note that $x_{t,i} + \alpha = s_t < 1$, so $\alpha < 1 - x_{t,i}$.
Let us define the function $g(y) = \frac{1}{\sqrt{x_{t,i} + y}} - \frac{x_{t,i}}{x_{t,i} + y}$ for $y \in [0, 1 - x_{t,i}]$. Notice that $ \frac{s_{t+1}}{s_t} = \sqrt{g(\alpha)} \ge \min_{y \in [0, 1 - x_{t,i}]} \sqrt{g(y)}$.
Let us differentiate $g(y)$ w.r.t. $y$, we get
\[
    g'(y) = \frac{-1}{2 (x_{t,i} + y)^{3/2}} + \frac{x_{t,i}}{(x_{t,i} + y)^2} = \frac{2 x_{t,i} - \sqrt{x_{t,i} + y}}{2 (x_{t,i} + y)^2}.
\]
As $x_{t,i} \le 1/4$, we have
\[
    x_{t,i} \le \frac{1}{4} \implies \sqrt{x_{t,i}} \le \frac{1}{2} \implies 2 x_{t,i} \le \sqrt{x_{t,i}} \implies 2 x_{t,i} - \sqrt{x_{t,i}} \le 0 \implies 2 x_{t,i} - \sqrt{x_{t,i} + y} \le 0,
\]
for every $y \ge 0$. Therefore, $g'(y) \le 0$ for very $y \ge 0$. So,
\[
    \min_{y \in [0, 1 - x_{t,i}]} g(y) = g(1 - x_{t,i}) = \frac{1}{\sqrt{x_{t,i} + (1 - x_{t,i})}} - \frac{x_{t,i}}{x_{t,i} + (1 - x_{t,i})} = 1 - x_{t,i} \ge \frac{3}{4},  
\]
as $x_{t,i} \le 1/4$. As $\sqrt{\cdot}$ is a monotone increasing function, so
\[
    \frac{s_{t+1}}{s_t} = \sqrt{g(\alpha)} \ge \min_{y \in [0, 1 - x_{t,i}]} \sqrt{g(y)} = \sqrt{g(1 - x_{t,i})} \ge \frac{\sqrt{3}}{2}.
\]

We now use this bound on $s_{t+1}$ to bound $z_{t+1}$. If $z_t = \ell$ for some $\ell \ge 1$, then $s_t \ge (1/2)^{2^{\ell-1}+1}$, which implies
\[
    s_{t+1} \ge \frac{\sqrt{3}}{2} s_t \ge \frac{\sqrt{3}}{2} \left(\frac{1}{2}\right)^{2^{\ell-1}+1} = \sqrt{3} \left(\frac{1}{2}\right)^{(2^{\ell-1}+1)+1} \ge \left(\frac{1}{2}\right)^{(2^{\ell-1}+1)+1} \ge \left(\frac{1}{2}\right)^{(2^{\ell})+1} \implies z_{t+1} \le z_t + 1.
\]

To summarize, we have (i) $z_{t+1} = \max(1, z_t-1)$ w.p. $\ge 1-p = 1-U$; (ii) $z_{t+1} \le z_t + 1$ w.p. $1$. So, $z_t$ has a (weakly) higher tendency of moving left than $y_t$, as required. Finally, let us upper bound $z_0 = y_0$ as a function $\gamma$. We claim that $k \le 1 + \lg\lg(1/\gamma)$ because
\begin{equation*}
    z_0 = k \le 1 + \lg\lg(1/\gamma) 
    \Longleftrightarrow s_0 \ge \left(\frac{1}{2}\right)^{2^{k-1}+1} 
    \impliedby s_0 \ge \left(\frac{1}{2}\right)^{\lg(\frac{1}{\gamma})+1} \ge \left(\frac{1}{2}\right)^{\lg(\frac{1}{\gamma})} = \gamma.
\end{equation*}


\end{proof}

%% file: proofs/part3.tex
\begin{proof}
We assume that the warm-up phase has been completed.
Let $f_i(x_{t})$ denote the value of the potential function at time $t+1$ given $i$ played at time $t$. If $i$ plays at time $t$, then $x_{t+1, i} = \sqrt{s_t - x_{t,i}} - (s_t - x_{t,i})$ and $s_{t+1} = \sqrt{s_t - x_{t,i}}$, and we have
\begin{align}
    f_i(x_{t}) - \frac{1}{6}\left(\frac{n-1}{n}\right)^3 &= \frac{1}{3} s_{t+1}^3 - \sum_{j < k} x_{t+1,j} x_{t+1,k} = \frac{1}{3} s_{t+1}^3 - x_{t+1,i} \sum_{j \neq i} x_{t+1,j} - \sum_{j < k, j \neq i, k \neq i} x_{t+1,j} x_{t+1,k} \nonumber \\
    &= \frac{1}{3} (s_t - x_{t,i})^{3/2} - (\sqrt{s_t - x_{t,i}} - (s_t - x_{t,i})) \sum_{j \neq i} x_{t+1,j} - \sum_{j < k, j \neq i, k \neq i} x_{t,j} x_{t,k} \nonumber \\
    &= \frac{1}{3} (s_t - x_{t,i})^{3/2} - (\sqrt{s_t - x_{t,i}} - s_t) \sum_{j \neq i} x_{t+1,j} - \sum_{j < k} x_{t,j} x_{t,k} \nonumber \\
    &= \frac{1}{3} (s_t - x_{t,i})^{3/2} - (\sqrt{s_t - x_{t,i}} - s_t) (s_t - x_{t,i})  - \sum_{j < k} x_{t,j} x_{t,k} \nonumber \\
    &= s_t(s_t - x_{t,i}) - \frac{2}{3} (s_t - x_{t,i})^{3/2} - \sum_{j < k} x_{t,j} x_{t,k}. \nonumber
\end{align}

We first show that $f_i(x_{t}) \le f(x_t)$ w.p. $1$ for every $i \in [n]$.
Let $g(y) = \frac{1}{3} s_t^3 - y^2 \left(s_t - \frac{2}{3} y \right)$ for $y \in [0,1]$. Notice that
\[
    f(x_t) - f_i(x_{t}) = \frac{1}{3} s_t^3 - \left( s_t(s_t - x_{t,i}) - \frac{2}{3} (s_t - x_{t,i})^{3/2} \right) = g(\sqrt{s_t - x_{t,i}}).
\]
Differentiating $g(y)$ w.r.t. $y$ we get $g'(y) = 2 y (y - s_t)$. As $g'(y) \le 0$ for $y \in [0,s_t]$ and $g'(y) \ge 0$ for $y \in [s_t,1]$, we have $\min_y g(y) = g(s_t) = \frac{1}{3} s_t^3 - s_t^2 \left(s_t - \frac{2}{3} s_t \right) = 0$. As $s_t - x_{t,i} \in [0,1]$, we have $g(\sqrt{s_t - x_{t,i}}) \ge g(s_t) = 0 \implies f(x_t) \ge f_i(x_{t})$, as required.

We next prove that the expected value of the potential $\Exp[f(x_t)]$ decreases by a multiplicative factor in any given time step $t$ to $t+1$ if $s_t \ge 1/4$. In particular, we show that $\Exp[f(x_{t+1})  | x_t] \le \left(1 - \kappa L \right) f(x_t)$ for some constant $\kappa > 0$, or equivalently, $\Exp[f(x_{t+1}) | x_t] - f(x_t) \le - \kappa L f(x_t)$, given $s_t \ge 1/4$. As $f_i(x_t)$ denotes the value of the potential at time $t+1$ given $i$ played at time $t$, so $\Exp[f(x_{t+1}) | x_t] = \sum_i w_{t,i}(x_t) f_i(x_t)$. Further, as $g(\sqrt{s_t - x_{t,i}}) = f(x_t) - f_i(x_t)$, we can alternatively prove the slightly stronger condition
\begin{multline}\label{eq:lm:part3:1}
    \Exp[f(x_{t+1}) | x_t] - f(x_t) \le -\kappa L f(x_t) 
    \Longleftrightarrow \sum_i w_{t,i}(x_t) (f(x_t) - f_i(x_t)) \\
    = \sum_i w_{t,i}(x_t) g(\sqrt{s_t - x_{t,i}}) \ge \kappa L f(x_t) 
    \impliedby \sum_i g(\sqrt{s_t - x_{t,i}}) \ge \kappa f(x_t).
\end{multline}

Let us assume that $s_t \ge 1/4$. Going back to our function $g(y) = \frac{1}{3} s_t^3 - y^2 \left(s_t - \frac{2}{3} y \right)$, we have $g''(y) = 4 y - 2 s_t$. We show that $g(y) \ge \kappa_g (y - s_t)^2$ for a positive constant $\kappa_g = \min\left( \frac{1}{20}, \left( \frac{1}{3} - \frac{27}{125} \right) \left( \frac{1}{4} \right)^3 \right) > 0$ using the case analysis below:
\begin{itemize}
    \item Case $y \ge 3s_t/5$. Then, $g''(y) = 4 y - 2 s_t \ge 2s_t/5 \ge 1/10$. As $g(s_t) = 0$ and $g'(s_t) = 0$, we get 
    \[
        g(y) \ge g(s_t) + g'(s_t) (y - s_t) + \frac{1}{2} \frac{2 s_t}{5} (y - s_t)^2 \ge \frac{1}{20} (y - s_t)^2.
    \]
    \item Case $y \le 3s_t/5$. As $g'(y) \le 0$ for $y \in [0, s_t]$, we have
    \[
        g(y) \ge g\left( \frac{3s_t}{5} \right) = \frac{1}{3} s_t^3 - \left( \frac{3s_t}{5} \right)^2 \left( s_t - \frac{2}{3}\frac{3s_t}{5} \right) = \left( \frac{1}{3} - \left(\frac{3}{5}\right)^3 \right) s_t^3 \ge \left( \frac{1}{3} - \frac{27}{125} \right) \left( \frac{1}{4} \right)^3.
    \]
    Now, as $y \in [0,1]$, we have $(y - s_t)^2 \le 1$, and we get our required result.
\end{itemize}
Plugging in $\sqrt{s_t - x_{t,i}}$ into $g(y)$, we get
\begin{equation}\label{eq:lm:part3:2}
    g(\sqrt{s_t - x_{t,i}}) \ge \kappa_g (\sqrt{s_t - x_{t,i}} - s_t)^2 \implies \sum_i g(\sqrt{s_t - x_{t,i}}) \ge \kappa_g \sum_i (\sqrt{s_t - x_{t,i}} - s_t)^2,
\end{equation}
given $s_t \ge 1/4$.
We now prove that $\sum_i (\sqrt{s_t - x_{t,i}} - s_t)^2 \ge \kappa_h f(x_t)$ for a positive constant $\kappa_h = \min\left( \left( \sqrt{\frac{6}{5}} - 1 \right)^2 \frac{1}{8}, \frac{1}{12} \right) > 0$, assuming $s_t \ge 1/4$. We do a case analysis on the value of $s_t$:
\begin{itemize}
    \item $\frac{1}{4} \le s_t \le \frac{5(n-1)}{6n}$. We claim that there exists an agent $i$ such that $\sqrt{s_t - x_{t,i}} \ge \sqrt{\frac{6}{5}} s_t$. Let this be false, i.e., for every $j \in [n]$, $\sqrt{s_t - x_{t,j}} < \sqrt{\frac{6}{5}} s_t$. Then, we get a contradiction because
    \begin{multline*}
        \sqrt{s_t - x_{t,j}} < \sqrt{\frac{6}{5}} s_t \Longleftrightarrow s_t - x_{t,j} < \frac{6}{5} s_t^2, \text{ for all $j \in [n]$} \\
        \implies \sum_j (s_t - x_{t,j}) < \sum_j \frac{6}{5} s_t^2 \implies (n-1) s_t < \frac{6}{5} n s_t^2 \implies s_t > \frac{5(n-1)}{6n}.
    \end{multline*}
    Now, as $\sqrt{s_t - x_{t,i}} \ge \sqrt{\frac{6}{5}} s_t$ for an agent $i$, we have
    \begin{equation}\label{eq:lm:part3:31}
        \sum_j (\sqrt{s_t - x_{t,j}} - s_t)^2 \ge (\sqrt{s_t - x_{t,i}} - s_t)^2 \ge \left( \sqrt{\frac{6}{5}} - 1 \right)^2 s_t^2 \ge \left( \sqrt{\frac{6}{5}} - 1 \right)^2 \frac{1}{8} f(x_t),
    \end{equation}
    because $s_t \ge 1/4$ and $f(x_t) \le 1/2$ (Lemma~\ref{lm:potential} \eqref{lm:potential:2}).

    \item $\frac{5(n-1)}{6n} \le s_t \le 1$. First, notice that
    \begin{multline}\label{eq:lm:part3:3}
        ((s_t - x_{t,i}) - s_t^2)^2 = (\sqrt{s_t - x_{t,i}} - s_t)^2 (\sqrt{s_t - x_{t,i}} + s_t)^2 \\
        \implies \sum_i (\sqrt{s_t - x_{t,i}} - s_t)^2 = \sum_i  \frac{((s_t - x_{t,i}) - s_t^2)^2}{(\sqrt{s_t - x_{t,i}} + s_t)^2} \ge \frac{1}{4} \sum_i ((s_t - x_{t,i}) - s_t^2)^2,
    \end{multline}
    because $s_t \le 1$ and $\sqrt{s_t - x_{t,i}} \le 1$. We do a change of variable. Let $y_k = s_t - x_{t,k}$. Let $\rho = \sum_k y_k = \sum_k (s_t - x_{t,k}) = (n-1) s_t \implies s_t = \rho / (n-1)$. Let us define the function $h(y) = \frac{1}{4} \sum_k (y_k - (\frac{\rho}{n-1})^2)^2 $. Taking derivative of $h$ we get
    \begin{align*}
        \frac{\partial h}{\partial y_i} &= \frac{1}{2} \left( y_i - \left( \frac{\rho}{n-1} \right)^2 \right) + \frac{1}{2} \sum_k \left( y_k - \left( \frac{\rho}{n-1} \right)^2 \right) \left( \frac{-2}{n-1} \frac{\rho}{n-1} \right) \\
        & \qquad = \frac{y_i}{2} - \frac{\rho^2}{2(n-1)^2} - \frac{\rho^2}{(n-1)^2} + \frac{n\rho^3}{(n-1)^4} = \frac{y_i}{2} + \frac{n\rho^3}{(n-1)^4} - \frac{3\rho^2}{2(n-1)^2}, \\
        \frac{\partial^2 h}{\partial y_i^2} &= \frac{1}{2} + \frac{3 n \rho^2}{(n-1)^4} - \frac{3 \rho}{(n-1)^2} , \qquad  \frac{\partial^2 h}{\partial y_i \partial y_j} = \frac{3 n \rho^2}{(n-1)^4} - \frac{3 \rho}{(n-1)^2}.
    \end{align*}
    So, the hessian $\nabla^2 h = \left( \frac{3 n \rho^2}{(n-1)^4} - \frac{3 \rho}{(n-1)^2} \right) + \frac{1}{2} I$, where $I$ is the identity matrix. Let's find the eigenvalues of $\nabla^2 h$. Using the matrix determinant lemma, we have
    \[
        \det(\nabla^2 h - \lambda I) = \det\left( \left( \frac{3 n \rho^2}{(n-1)^4} - \frac{3 \rho}{(n-1)^2} \right) + \left( \frac{1}{2} - \lambda \right) I \right) = \left( 1 + \frac{n  \left( \frac{3 n \rho^2}{(n-1)^4} - \frac{3 \rho}{(n-1)^2} \right)}{\left( \frac{1}{2} - \lambda \right)}\right) \left( \frac{1}{2} - \lambda \right)^n.
    \]
    Setting $\det(\nabla^2 h - \lambda I) = 0$ we get $\lambda = \frac{1}{2}$ or $\lambda = \frac{1}{2} + \frac{3 n^2 \rho^2}{(n-1)^4} - \frac{3 n \rho}{(n-1)^2} = \frac{1}{2} + 3\frac{np}{(n-1)^2} \left(\frac{np}{(n-1)^2} - 1 \right)$. As $s_t \ge \frac{5(n-1)}{6n} \implies \rho = (n-1) s_t \ge \frac{5(n-1)^2}{6n} \implies \frac{np}{(n-1)^2} \ge \frac{5}{6}$, we get $\lambda \ge \frac{1}{2} + 3 \frac{5}{6} \left( \frac{5}{6} - 1 \right) = \frac{1}{12}$. So, $h$ is $(1/12)$-strongly convex. Now, let $y_k^* = \frac{(n-1)^2}{n^2}$ for every $k \in [n]$, which corresponds to the equilibrium. Notice that we have $\rho^* = \frac{(n-1)^2}{n} \ge \frac{5(n-1)^2}{6n}$. Further, $h(y^*) = 0$ because $\left( y_i^* - \left( \frac{\rho}{n-1} \right)^2 \right)^2 = 0$ for very $i \in [n]$ and $\nabla h(y^*) = 0$ because
    \begin{align*}
        (\nabla h(y^*))_i = \frac{y_i^*}{2} + \frac{n(\rho^*)^3}{(n-1)^4} - \frac{3(\rho^*)^2}{2(n-1)^2} = \frac{1}{2} \frac{(n-1)^2}{n^2} + \frac{n}{(n-1)^4} \frac{(n-1)^6}{n^3} - \frac{3}{2(n-1)^2} \frac{(n-1)^4}{n^2} = 0,
    \end{align*}
    for every $i \in [n]$. So, we can write
    \[
        h(y) \ge h(y^*) + (\nabla h(y^*))^{\intercal} (y - y^*) + \frac{1}{24} || y - y^* ||_2^2 = \frac{1}{24} || y - y^* ||_2^2.
    \]
    From Lemma~\ref{lm:potential} \eqref{lm:potential:5}, we know that $|| y - y^* ||_2^2 \ge 2 f(x_t)$. Putting everything together, 
    \begin{equation}\label{eq:lm:part3:32}
        \sum_j (\sqrt{s_t - x_{t,j}} - s_t)^2 \ge \frac{1}{4} \sum_j ((s_t - x_{t,j}) - s_t^2)^2 = h(y) \ge \frac{1}{24} || y - y^* ||_2^2 \ge \frac{1}{12} f(x_t).
    \end{equation}
\end{itemize}
Combining \eqref{eq:lm:part3:1}, \eqref{eq:lm:part3:2}, \eqref{eq:lm:part3:31}, and \eqref{eq:lm:part3:32}, we have $\Exp[f(x_{t+1}) | x_t] \le \left( 1 - \kappa_g \kappa_h L \right) f(x_t)$ assuming $s_t \ge 1/4$. 

Let $\kappa = \kappa_g \kappa_h$.
Taking expectation w.r.t. $x_t$, we get $\Exp[f(x_{t+1})] \le \left( 1 - \kappa L \right) \Exp[f(x_t)]$ for all $x_t$ conditioned on $s_t \ge 1/4$. 
Further, as $f_i(x_t) \le f(x_t)$ for every $i \in [n]$ w.p. $1$, we also have $\Exp[f(x_{t+1})] \le \Exp[f(x_t)] $ w.p. $1$.
So, the sequence $\Exp[f(x_t)]$ is a non-increasing sequence and decreases by a constant factor for time steps when $s_t \ge 1/4$.

Assuming that there have been at least $m \ge \frac{1}{\kappa L} \ln\left( \frac{1}{2 \epsilon \delta} \right)$ steps after the warm-up phase until time $t$ where the total output was at least $1/4$, we get
\[
    \Exp[f(x_t)] \le (1 - \kappa L)^m f(x_{T_{warm}}) \le \frac{(1 - \kappa L)^m}{2}  \le \frac{e^{-\kappa L m}}{2} \le \epsilon \delta.
\]
As $\Exp[f(x_{t})] \le \epsilon \delta$ and $f(x_{t}) \ge 0$, using Markov inequality we have $\Prob[f(x_{t}) \ge \epsilon] \le \frac{\Exp[f(x_{t})]}{\epsilon} \le \delta$, which implies $\Prob[f(x_{t}) \le \epsilon]$ w.p. $1-\delta$, as required.
\end{proof}

%% file: 7simulation.tex
\section{Simulations}
For $2$ homogeneous agents, our convergence bound for the best-response dynamics is tight up to a constant number of time steps; our simulations reflect the same. Here, we focus on $n \ge 3$ homogeneous agents where our lower and upper bounds do not match; particularly, w.r.t. the parameter $\epsilon$ we have an exponential gap between $\log\log(1/\epsilon)$ and $\log(1/\epsilon)$.

We plot the convergence times for five models for selecting the agent at each time step: (i) the agent is selected uniformly at random (\unif), for which we provide bounds in Corollary~\ref{thm:hom:unif}; (ii) the agents are selected in a round robin fashion (\round); (iii) the lexicographically smallest agent whose utility increases by more than $\epsilon$ by playing the best response is selected (\lex); (iv) a myopic-worst-case model that selects the agent whose utility increases by the smallest possible value above $\epsilon$ (\worst); (v) a myopic-best-case model that selects the agent whose utility increases by the largest possible value (\best). We start from the initial profile of $(a, 0, \ldots, 0)$.\footnote{The results were similar for a random starting profile. Moreover, the dependency on the starting profile, $\gamma$, is very weak.} We plot the time till convergence to an $\epsilon$-approximate equilibrium with respect to one of three parameters (while keeping the others fixed): $\epsilon$, $\gamma \equiv a$, and $n$. When these parameters are not being varied, we set them to $\epsilon = 10^{-10}$, $\gamma = 10^{-10}$, and $n = 10$. We took an average over $100$ runs for \unif. 
The omitted plots are in Appendix~\ref{sec:app:plots}.

We observe the following rates of convergence w.r.t. $\epsilon$; for $n$ and $\gamma$, see Appendix~\ref{sec:app:plots}. The time for \unif, \round\, and \best\ seems proportional to $\log(1/\epsilon)$; we could not pin down the exact slope for \lex\ and \worst, but seems somewhere between $(\log(1/\epsilon))^5$ and $(1/\epsilon)^{1/5}$. Important observations: (i) it is unlikely that we can get a $\log\log(1/\epsilon)$ bound for \unif, \round, and \best; (ii) \lex\ and \worst\ may have a polynomial rather than poly-logarithmic dependency on $1/\epsilon$. Moreover, observation (ii) also highlights that the worst-case convergence time may be $\text{poly}(1/\epsilon)$ instead of $\text{polylog}(1/\epsilon)$.

\begin{figure*}[ht]
  \centering
  \subfloat[Time vs $(1/\epsilon)$]{
  \includegraphics[width=0.5\textwidth]{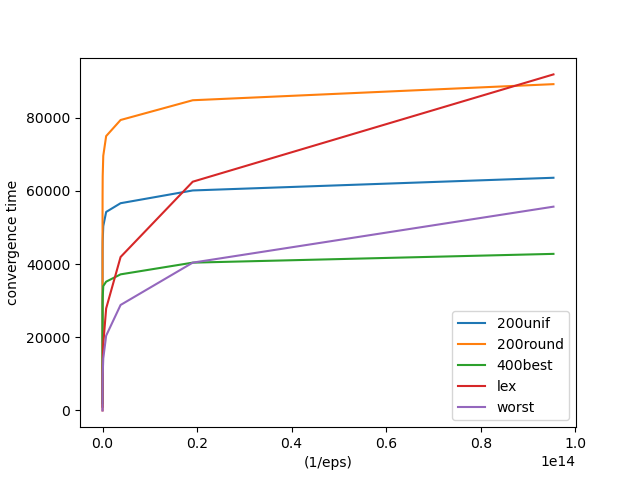}
  }
  \subfloat[Time vs $\log(1/\epsilon)$]{
  \includegraphics[width=0.5\textwidth]{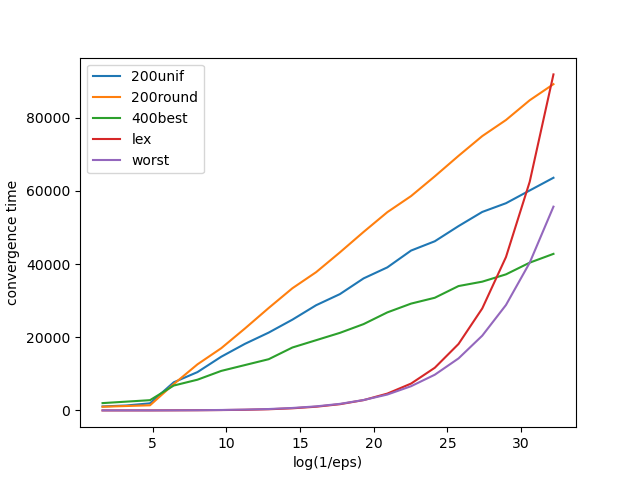}
  }
\end{figure*}

%% file: 8conclude.tex
\section{Conclusion and Open Problems}
We showed fast convergence of best-response dynamics in lottery contests with homogeneous agents to an $\epsilon$-approximate equilibrium. We get the following dependency on relevant parameters: (i) tight double-logarithmic dependency of $\Theta(\log\log(1/\epsilon))$ for two agents and logarithmic upper bound of $O(\log(1/\epsilon))$ for more than two agents; (ii) almost linear dependency of $\tilde{\Theta}(n)$ on the number of agents; (iii) double-logarithmic dependency on the initial state, captured by the $\Theta(\log\log(1/\gamma))$ term. For non-homogeneous agents, we show non-convergence using best-response cycles.

The following are a few related open directions:
\begin{enumerate}
    \item For more than two agents, our lower bound $\Omega(\log\log(1/\epsilon))$ and upper bound $O(\log(1/\epsilon))$ with respect to $\epsilon$ has an exponential gap. Our simulations indicate a $\Omega(\log(1/\epsilon))$ lower bound. Prove this formally.
    \item For almost homogeneous agents, our simulations show convergence to the equilibrium. Characterize the amount of non-homogeneity that allows convergence (or causes non-convergence). See Section~\ref{sec:non-hom} for more details.
    \item Tullock contests with convex cost functions, which includes Tullock contests with concave contest success functions~\cite{tullock1980efficient}, have unique pure-strategy Nash equilibrium. Prove convergence of best-response dynamics for these models for homogeneous agents (and non-convergence for non-homogeneous agents).
\end{enumerate}

%% file: 9appendix.tex
\section{Omitted Proofs}\label{sec:app:proofs}
\input{proofs/thm-2hom.tex}

\input{proofs/coupon.tex}
\input{proofs/prop1.tex}
\input{proofs/part1.tex}
\input{proofs/outputlb.tex}

\input{proofs/markov.tex}

\input{proofs/potential.tex}
\input{proofs/lipschitz.tex}

\section{All Simulation Plots}\label{sec:app:plots}
We observe the following w.r.t. $n$ and $\gamma$.
\begin{enumerate}
    \item $n$: The time for \unif\ and \best\ is proportional to $n \log n$ and $n$. The time for \round, \lex, and \worst\ seems between $n^2$ and $n^3$. Important observations: (i) our bound for \unif\ is tight; (ii) we have a polynomial dependency on $n$ for all models.
    \item $\gamma$: $\gamma$ only has an additive $\Theta(\log\log(1/\gamma))$ term for \unif\ (Corollary~\ref{thm:hom:unif}) without a function of $n$ or $\epsilon$ being multiplied to $\log\log(1/\gamma)$. Our simulations show a similar dependency for \lex, \round, \worst, and \best.
\end{enumerate}
\begin{figure}[ht]
    \centering
    \subfloat[Time vs $\log\log(1/\epsilon)$]{
        \includegraphics[width=0.45\textwidth]{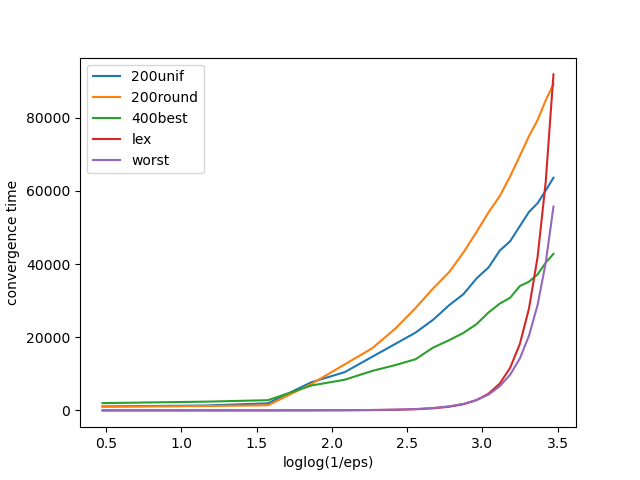}
        \label{fig:eps:3}
    }\\
    \subfloat[Time vs $(1/\epsilon)^{1/5}$]{
        \includegraphics[width=0.45\textwidth]{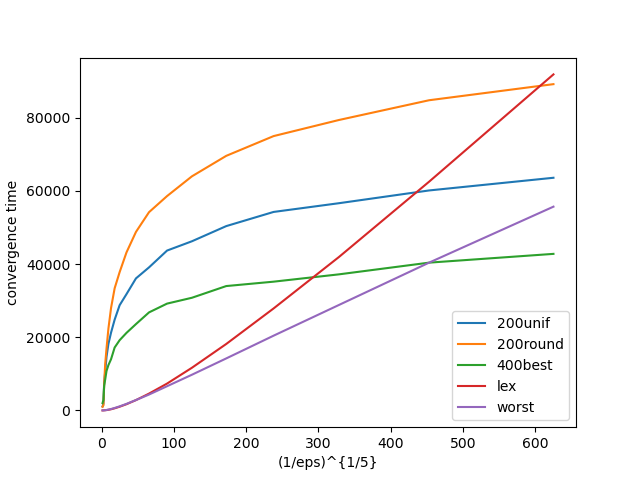}
        \label{fig:eps:4}
    }
    \subfloat[Time vs $(\log(1/\epsilon))^5$]{
        \includegraphics[width=0.45\textwidth]{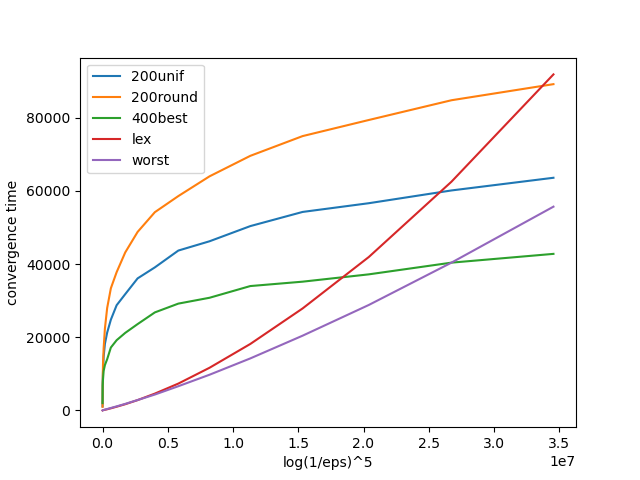}
        \label{fig:eps:5}
    }
\caption{Convergence time w.r.t. $\epsilon$.}
\label{fig:eps}
\end{figure}

\begin{figure}[ht]
  \centering
  \subfloat[Time vs $n \log n$]{
  \includegraphics[width=0.45\textwidth]{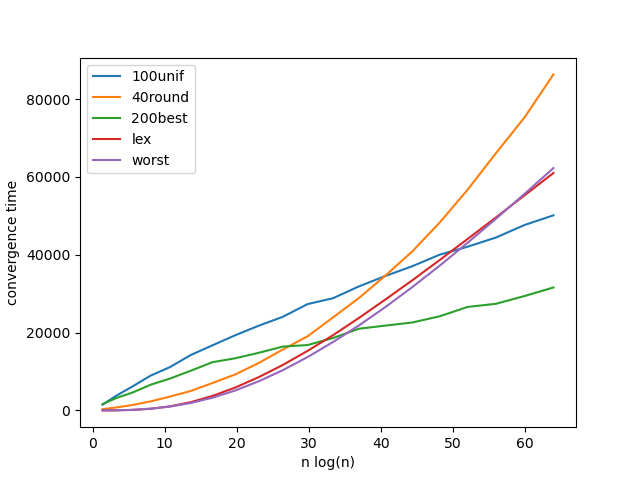}
  \label{fig:n:1}
  }\\
  \subfloat[Time vs $n^2$]{
  \includegraphics[width=0.45\textwidth]{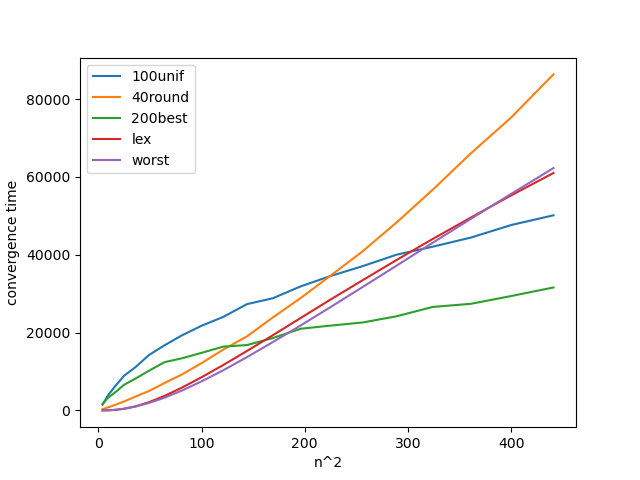}
  \label{fig:n:2}
  }
  \subfloat[Time vs $n^3$]{
  \includegraphics[width=0.45\textwidth]{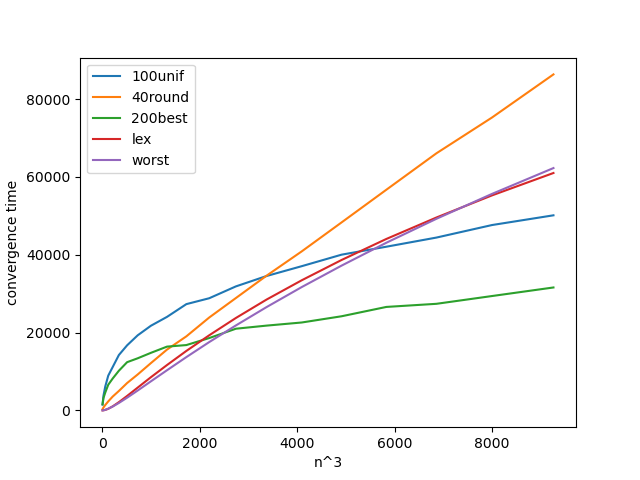}
  \label{fig:n:3}
  }
\caption{Convergence time w.r.t. $n$.}
\label{fig:n}
\end{figure}

\begin{figure}[ht]
  \centering
  \subfloat[Time vs $\log\log(1/\gamma)$.]{
  \includegraphics[width=0.45\textwidth]{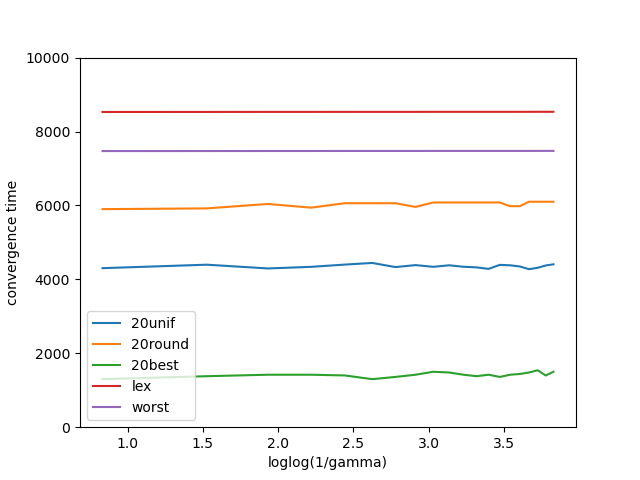}
  \label{fig:gamma:1}
  }
  \subfloat[Relative convergence time (after subtracting the time taken for $\gamma = 1/10$ for each algorithm) vs $1/\epsilon$.]{
  \includegraphics[width=0.45\textwidth]{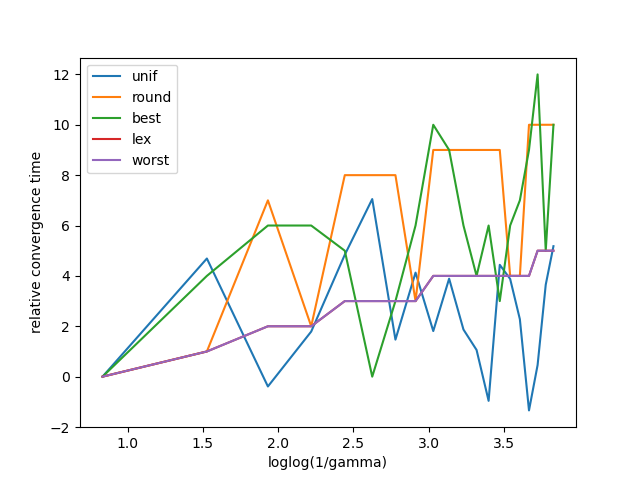}
  \label{fig:gamma:2}
  }
\caption{Convergence time w.r.t.  $\gamma$.}
\label{fig:gamma}
\end{figure}




%% file: proofs/thm-2hom.tex
\begin{proof}[Proof of Theorem~\ref{thm:hom2}]
We use the rate of convergence of the sequence $z_t$ in Lemma~\ref{lm:2hom} to prove fast convergence of best-response dynamics in lottery contests. We do so by mapping the state (action profile) $x_t$  at time $t$ of the best-response dynamics to the state $z_t$ of the sequence.

Notice that only the action of agent $1$ in the initial state $x_0 = (x_{0,1}, x_{0,2})$ is relevant for best-response dynamics, because, at $t=0$, agent $2$ makes the best-response move, which is a function of only $x_{0,1}$, and which overwrites $x_{0,2}$. Also, notice that w.l.o.g. we can assume that $0 < x_{0,1} < 1/4$ because even if it is not true, we get it in a couple of steps:
\begin{enumerate}
    \item \label{lm:2agents:1} If $x_{0,1} = 0$, then we know that $x_{1,2} = a \in (0, 1/4)$ by our assumptions on $a$. We shift the timeline by $1$ time step and get the required condition.
    \item \label{lm:2agents:2} If $x_{0,1} \ge 1$, then we get $x_{1,2} = 0$ and we get case \eqref{lm:2agents:1} in one additional time step.
    \item If $x_{0,1} = 1/4$, then we are already at the equilibrium.
    \item If $\frac{1}{4} < x_{0,1} < 1$, then in one step we have $x_{1,2} = \sqrt{x_{0,1}} - x_{0,1} \in (0, 1/4)$ because $\sqrt{x_{0,1}} - x_{0,1} > 0$ for $x_{0,1} \in (0,1)$ and $\sqrt{x_{0,1}} - x_{0,1} \le \max_{y \in (0,1)} y(1-y) < 1/4$.
\end{enumerate}

Let us define the sequence $z_t$ as follows:
\[
    z_t =
    \begin{cases}
        \sqrt{x_{t,1}} &\text{ if $\bmod(t,2) = 0$,} \\
        \sqrt{x_{t,2}} &\text{ if $\bmod(t,2) = 1$.}
    \end{cases}
\]
$z_t$ has a one-to-one correspondence with the states of the best-response dynamics of the lottery contest. If $\bmod(t,2) = 0$, then $x_{t+1,2} = \sqrt{x_{t,1}} - x_{t,1} \Longleftrightarrow z_{t+1}^2 = z_{t} - z_t^2 \Longleftrightarrow z_{t+1} = \sqrt{z_t (1 - z_t)}$. An analogous equivalence holds for the case $\bmod(t,2) = 1$ as well. Notice that we can assume $z_0 = \sqrt{x_{0,1}} \in (0, 1/2)$ as per our discussion in the last paragraph. 

From Lemma~\ref{lm:2hom}, we know that in $\lg\lg(\frac{4}{\epsilon}) + \lg\lg(\frac{1}{\gamma}) + \Theta(1) = \lg\lg(\frac{1}{\epsilon}) + \lg\lg(\frac{1}{\gamma}) + \Theta(1)$ steps $z_t$ will reach $\ge \frac{1}{2} - \frac{\epsilon}{16}$ starting from $z_0 = \gamma$. As per our definition of $z_t$ (based on $x_t$), we can set $\gamma = \sqrt{x_{0,1}}$ if $0 < x_{0,1} < \frac{1}{4}$, $\gamma = x_{0,1} - \sqrt{x_{0,1}}$ if $\frac{1}{4} \le x_{0,1} < 1$, and $\gamma = \sqrt{a}$ otherwise. Let $\hepsilon = \frac{\epsilon}{4}$. In the rest of the proof, we show that $z_{t-1} \ge \frac{1}{2} - \hepsilon$ implies $x_t$ is an $\epsilon$-approximate equilibrium. We can assume $\epsilon \in [0,1] \Longleftrightarrow \hepsilon \in [0, \frac{1}{16}]$.

Let's say $t$ is odd. As $z_{t-1} \ge \frac{1}{2} - \hepsilon$, therefore $x_{t,1} = x_{t-1, 1} = z_{t-1}^2 \ge (\frac{1}{2} - \hepsilon)^2 \ge \frac{1}{4} - \hepsilon$. Also, $x_{t,2} = z_{t}^2 > z_{t-1}^2 \ge \frac{1}{4} - \hepsilon$. A similar argument applies when $t$ is even. So, $\frac{1}{4} - \hepsilon \le x_{t,i} \le \frac{1}{4}$ for $i \in \{1,2\}$.

Let us first lower bound the utility that an agent, w.l.o.g. agent $1$, gets at the profile $x_{t}$ given $x_{t,1}, x_{t,2} \in \frac{1}{4} - [0, \hepsilon]$.
\begin{align*}
	u_1(x_{t,1}, x_{t,2}) = \frac{x_{t,1}}{x_{t,1} + x_{t,2}} - x_{t,1} \ge  \frac{x_{t,1}}{x_{t,1} + \frac{1}{4}} - x_{t,1} = \frac{4x_{t,1}}{4x_{t,1} + 1} - x_{t,1}.
\end{align*}
Let $g(y) = \frac{4y}{4y + 1} - y$ for $y \in [0, \frac{1}{4}]$. The derivative is $g'(y) = \frac{4}{(4y + 1)^2} - 1 = \frac{4 - (4y+1)^2}{(4y+1)^2}$. As $y \in [0, \frac{1}{4}]$, therefore $(4y + 1)^2 \in [1, 4]$, so $g'(y) \ge 0$. As $x_{t,1} \ge \frac{1}{4} - \hepsilon$, this implies that 
\begin{align*}
	u_1(x_{t,1}, x_{t,2}) \ge \frac{\frac{1}{4} - \hepsilon}{\frac{1}{4} - \hepsilon + \frac{1}{4}} - \left(\frac{1}{4} - \hepsilon\right) = \frac{\frac{1}{2} - 2\hepsilon}{1 - 2\hepsilon} - \left(\frac{1}{4} - \hepsilon\right) \ge \left(\frac{1}{2} - 2\hepsilon\right) - \left(\frac{1}{4} - \hepsilon\right) = \frac{1}{4} (1 - 4\hepsilon)
\end{align*}
Let us now upper bound the utility that agent $1$ can get by deviating.
\begin{multline*}
	\max_s u_1(s, x_{t,2}) = \max_s \frac{s}{s + x_{t,2}} - s \\
	= \frac{\sqrt{x_{t,2}} - x_{t,2}}{(\sqrt{x_{t,2}} - x_{t,2}) + x_{t,2}} - (\sqrt{x_{t,2}} - x_{t,2}) = 1 - \sqrt{x_{t,2}} - (\sqrt{x_{t,2}} - x_{t,2}) = (1 - \sqrt{x_{t,2}})^2.
\end{multline*}
As $(1 - \sqrt{x_{t,2}})^2$ increases as $x_{t,2}$ decreases, and $x_{t,2} \ge \frac{1}{4} - \hepsilon$, so
\begin{multline*}
	\max_s u_1(s, x_{t,2}) \le \left( 1 - \sqrt{\frac{1}{4} - \hepsilon} \right)^2 = \frac{1}{4} \left( 2 - \sqrt{1 - 4\hepsilon} \right)^2 \le \frac{1}{4} \left( 2 - (1 - 4\hepsilon) \right)^2 = \frac{1}{4} \left( 1 + 4 \hepsilon \right)^2 \\
	= \frac{1}{4} ( 1 + (16 \hepsilon) \hepsilon + 8 \hepsilon ) \le \frac{1}{4} ( 1 + 9 \hepsilon ) \text{ as $\hepsilon \le \frac{1}{16}$}.
\end{multline*}
Finally, to prove that $x_t$ is an $\epsilon$-approximate equilibrium we need
\begin{multline*}
	u_1(x_{t,1}, x_{t,2}) \ge (1 - \epsilon) \max_s u_1(s, x_{t,2}) \impliedby \frac{1}{4} (1 - 4\hepsilon) \ge (1 - \epsilon) \frac{1}{4} ( 1 + 9 \hepsilon ) \Longleftrightarrow 1 - \epsilon \le \frac{1 - 4 \hepsilon}{1 + 9 \hepsilon} \\
	\impliedby 1 - \epsilon \le (1 - 4\hepsilon) (1 - 9\hepsilon) \impliedby 1 - \epsilon \le 1 - 13 \hepsilon \Longleftrightarrow \epsilon \ge 13 \hepsilon \impliedby \epsilon = 16 \hepsilon.
\end{multline*}
\end{proof}

%% file: proofs/coupon.tex
\begin{proof}[Proof of Lemma~\ref{lm:coupon}]
Let $\tau$ be the time it takes to collect all $n$ coupons in the coupon collector problem where each coupon is selected w.p. $1/n$, the following result is well-known.
\begin{lemma}\cite{mitzenmacher2017probability}\label{lm:coupon:unif}
    For any constant $c > 0$, we have the following high probability bounds: (i) upper bound, $\Prob[\tau > n \ln n + cn] < e^{-c}$; (ii) lower bound, $\Prob[\tau < n \ln n - cn] < e^{-c}$.
\end{lemma}
The analysis underlying Lemma~\ref{lm:coupon:unif} goes as follows: The time to collect all coupons $\tau$ can be decomposed as $\tau = \sum_i \tau_i$, where $\tau_i$ is the time it takes to collect the $i$-th coupon. After $(i-1)$ coupons have been collected, the probability that we get a coupon that has not yet been collected in the next time step is equal to $p_i = \frac{n - i + 1}{n}$. So, $\tau_i$ corresponds to the time till the first head of a geometric random variable with parameter $p_i$.  In particular, $p_1 = 1$, $p_2 = (n-1)/n$, $p_3 = (n-2)/n$, and so on.

For the best-response dynamics, let $T = \sum_i T_i$ denote the time it takes to for every agent to play at least once, where $T_i$ is the time between the $(i-1)$-th and the $i$-th unique agent. Remember that we are doing a worst-case analysis over the random selection process parameterized by $L$ and $U$. In the first time step, we get the first unique agent w.p. $1 \ge L n p_1$ as $L \le 1/n$. In the second time step, the maximum possible probability that the selection process can assign to already selected agents is bounded above by $\min(U
, 1 - (n-1)L) \le 1 - (n-1)L$. So, the probability of selecting a new agent is at least $(n-1) L = \frac{(n-1) L n}{n} = L n p_2$. Similarly, we can show that for all $i$, after exactly $(i-1)$ agents have played at least once, the probability that we select an agent who has not yet played in the next time step is at least $L n p_i$. This implies that $\Prob[T_i \le k] \le Ln \Prob[\tau_i \le k]$ for all $i \in [n]$ and $k \ge 1$. Using the upper bound in Lemma~\ref{lm:coupon:unif}, we get $\Prob[T > \frac{1}{L}(\ln n + c)] < e^{-c}$, setting $c = \ln(1/\delta)$, we get $\Prob[T \le \frac{1}{L}(\ln n + \ln (\frac{1}{\delta}))] \ge 1 - \delta$ as required.

The idea for the lower bound is similar. As discussed earlier, after exactly $(i-1)$ unique agents have played, the total probability that the selection process can assign to the agents who have already played, in worst-case, is equal to $\min((i-1)U, 1 - (n-(i-1))L)$. If $(i-1)U \ge 1$, then $\min((i-1)U, 1 - (n-(i-1))L) = 1 - (n-(i-1))L$, which happens for at least the last $n - \frac{1}{U}$ agents. Moreover, the time it takes to collect the later coupons is higher than the earlier ones (probabilistically). So, using the lower bound of Lemma~\ref{lm:coupon:unif}, we get $\Prob[T \ge \frac{1}{L}(1 - \frac{1}{nU}) \log(n \delta)] \ge 1 - \delta$. On the other hand, the selection process can select agents uniformly, which gives us the bound $\Prob[T \ge n \log(n \delta)] \ge 1 - \delta$. So, $\Prob[T \ge \max(\frac{1}{L}(1 - \frac{1}{nU}), n) \log(n \delta)] \ge 1 - \delta$ in the worst case.
\end{proof}

%% file: proofs/prop1.tex
\begin{proof}[Proof of Lemma~\ref{lm:prop1}]
    Let's first show that $s_t > 0$ for $t \ge 1$. For contradiction, say $s_t = 0$ for some $t \ge 1$. Let $i = i_{t-1}$ denote the agent that moved at time $t-1$. Notice that $s_t = 0 \implies s_{t-1} - x_{t-1,i} = s_{t} - x_{t} \le s_t = 0$. If $s_{t-1} - x_{t-1,i} = 0$, then as a best response to this, agent $i$ plays $x_{t,i} = a > 0$. So, $s_t > 0$, contradiction.

    Now, for any $t$, if $s_t < 1$, then $s_{t + 1} < 1$ because
    \begin{itemize}
        \item if $s_{t} - x_{t, i_{t}} = 0$, then $s_{t + 1} = x_{t+1, i_{t}} = a < 1$,
        \item if $0 < s_{t} - x_{t, i_{t}} < 1$, then $s_{t + 1} = (s_{t} - x_{t, i_{t}}) + x_{t+1, i_{t}} = (s_{t} - x_{t, i_{t}}) + \sqrt{s_{t} - x_{t, i_{t}}} - (s_{t} - x_{t, i_{t}}) = \sqrt{s_{t} - x_{t, i_{t}}} < 1$,
        where $x_{t+1, i_{t}} = \sqrt{s_{t} - x_{t, i_{t}}} - (s_{t} - x_{t, i_{t}})$ is the best response of agent $i_{t}$ at time $t$ (see equation \eqref{eq:br:single}).
    \end{itemize}

    Finally, let us assume that at time $t$, $x_{t,i} > 0$, $x_{t,j} > 0$, and $s_t < 1$.  As $s_{t} - x_{t,i} \ge x_{t,j} > 0$ and $s_{t} - x_{t,i} < s_{t} < 1$, so $x_{t+1,i} = \sqrt{s_{t} - x_{t,i}} (1 - \sqrt{s_{t} - x_{t,i}}) > 0$. Similarly, $x_{t+1,j} > 0$.
\end{proof}

%% file: proofs/part1.tex
\begin{proof}[Proof of Lemma~\ref{lm:part1}]
Let $T_1$ be a random variable that denotes the time it takes for every agent to make at least one move. We next prove that conditions \eqref{def:warmup:1} and \eqref{def:warmup:3} of the warm-up phase (Definition~\ref{def:warmup}) are satisfied for $t \ge T_1$.

Condition \eqref{def:warmup:3}.
Notice that Lemma~\ref{lm:prop1} \eqref{lm:prop1:1} implies $s_t > 0$ for $t \ge T_1 \ge n \ge 1$.
Let us now prove that $s_t < 1$ for $t \ge T_1$. Note that it is enough to prove $s_{T_1} < 1$ because $s_{T_1} < 1$ implies $s_t < 1$ for all $t \ge T_1$ using Lemma~\ref{lm:prop1} \eqref{lm:prop1:2}.
For contradiction, let us assume that $s_{T_1} \ge 1$. Let $t = T_1 - 1$. 
\begin{itemize}
    \item If $s_t - x_{t, i_t} = 0$, then $s_{T_1} = s_{t+1} = x_{t+1,i_t} = a < 1$, which contradicts $s_{T_1} \ge 1$.
    \item If $0 < s_t - x_{t, i_t} < 1$, then $s_{T_1} = s_{t+1} = \sqrt{s_t - x_{t, i_t}} < 1$, which also leads to contradiction. 
    \item If $s_t - x_{t, i_t} \ge 1$, then $x_{t+1,i_t} = 0$ and $s_t \ge s_t - x_{t, i_t} \ge 1$. Doing induction in the backward direction, starting from $\tau = t = T_1 - 1$ and going to $\tau = 0$, $x_{\tau+1,i_\tau} = 0$ for every $\tau < T_1$. As every agent $i$ makes at least one move before $T_1$, so every agent $i \in [n]$ has $x_{T_1, i} = 0$. So, $s_{T_1} = 0$. Contradiction. 
\end{itemize}
This completes the proof for condition \eqref{def:warmup:3} of the warm-up phase.

Condition \eqref{def:warmup:1}.
Let us now prove that $x_{t,i} \le 1/4$ for all $t \ge T_1$ and $i \in [n]$. Fix a $t \ge T_1$ and an $i \in [n]$. As every agent has played at least once before $T_1$, so $i$ must have made a move before $t$. Let $\tau < t$ denote the most recent time before $t$ when agent $i$ made the best-response move. 
By definition of $\tau$, $i$ did not play between $\tau+1$ and $t$, so $x_{t,i} = x_{\tau+1,i}$. Finally, $x_{t,i} \le 1/4$ because
\begin{itemize}
    \item if $s_{\tau} - x_{\tau,i} = 0$, then $x_{\tau+1,i} = a \le 1/4$,
    \item if $0 < s_{\tau} - x_{\tau,i} < 1$, then $x_{\tau+1,i} = \sqrt{s_{\tau} - x_{\tau,i}} - (s_{\tau} - x_{\tau,i}) \le \max_{\beta \in (0,1)} \beta (1 - \beta) \le 1/4$, and
    \item if $s_{\tau} - x_{\tau,i} \ge 1$, then $x_{\tau+1,i} = 0 \le 1/4$.
\end{itemize}
This completes the proof for condition \eqref{def:warmup:1} of the warm-up phase.

We have shown that $s_t > 0$ for $t \ge T_1$, so there must be at least one agent with positive output for all $t \ge T_1$. Let $i$ be an agent that has positive output $x_{T_1,i} > 0$ at time $T_1$. But, it is possible that $i$ is the only agent with positive output, and every other agent $j \neq i$ has $x_{T_1,j} = 0$.

Condition \eqref{def:warmup:2}.
Let $T_2$ denote the additional steps after $T_1$, if any, required to get at least two agents with positive output. Notice that $T_2$ can be upper bounded by the time it takes to get the first head (select an agent $j \neq i$) of a geometric random variable with parameter $p \ge (1 - U)$ because $i$ can be assigned a probability of at most $U$ at each time step. When a $j \neq i$ is selected at time $t = T_2 - 1$, then $x_{T_2,j} = x_{t+1,j} = \sqrt{s_t - x_{t,j}} - (s_t - x_{t,j}) = \sqrt{x_{t,i}} - x_{t,i} \in (0,1)$ as $x_{t,i} \in (0,1)$. Further, from Lemma~\ref{lm:prop1} \eqref{lm:prop1:3}, we know that there will always be at least two agents with positive output this time onward. So, for $t \ge T_1 + T_2$, the action profile $x_t$ satisfies condition \eqref{def:warmup:2} of the warm-up phase.

To summarize, for $t \ge T_1 + T_2$, the action profile $x_t$ satisfies all conditions required for the completion of the warm-up phase (Definition~\ref{def:warmup}). So, $T_{warm} \le T_1 + T_2$. Let us now prove a high probability upper bound on $T_1 + T_2$, which implies the same bound for $T_{warm}$. 
From Lemma~\ref{lm:coupon}, we know that $\Prob[T_1 > \frac{1}{L}\ln (\frac{n}{\delta})] < \delta$ for any $\delta \in (0,1)$. Set $\delta = \frac{\delta_1}{2}$, we get $\Prob[T_1 > \frac{1}{L}\ln (\frac{2n}{\delta_1})] < \frac{\delta_1}{2}$.
As $T_2$ underestimates the time till the first head for a geometric random variable with parameter $(1-U) > \frac{1}{2}$, we have $\Prob[T_2 > k] \le U^k < (\frac{1}{2})^k$. Setting $k = \lg(\frac{2}{\delta_1})$, we get $\Prob[T_2 > \lg(\frac{2}{\delta_1})] < \frac{\delta_1}{2}$.
Using union bound, we get $T_{warm} \le T_1 + T_2 \le \frac{1}{L}\ln (\frac{2n}{\delta_1}) + \lg(\frac{2}{\delta_1}) = O(\frac{1}{L}\log(\frac{n}{\delta_1}))$ w.p. $1- \delta_1$ as required.
\end{proof}

%% file: proofs/outputlb.tex
\begin{proof}[Proof of Lemma~\ref{lm:outputlb}]
If $T_{warm} = 0$, i.e., all required conditions for completion of the warm-up phase are satisfied by the initial state $x_0$, then we trivially get $s_{T_{warm}} = s_0 \ge \gamma$.

Let us assume that $T_{warm} > 0$. Let us focus on the time step $T_{warm} - 1$. Let $t = T_{warm} - 1$. As $T_{warm} = t+1 > 0$, $t \ge 0$. We need to lower bound $s_{t+1}$.

As $t+1$ is the smallest time when all conditions for completion of the warm-up phase, Definition~\ref{def:warmup}, are satisfied, therefore at time $t$, there must be at least one violation. We do a case analysis depending upon which condition was violated at $t$.

\noindent\textbf{Case 1:} Definition~\ref{def:warmup} condition \eqref{def:warmup:2} violated at $t$, i.e., there is only one agent $i$ with $x_{t,i} > 0$. 

First, notice that at time $t$ an agent $j \neq i$ makes a transition to a positive $x_{t+1,j}$ to satisfy all three conditions of Definition~\ref{def:warmup}. Check that the conditions \eqref{def:warmup:1} and \eqref{def:warmup:3} of Definition~\ref{def:warmup} must not have been violated at time $t$ because then in a single step, we could not have satisfied all three conditions. This implies that $s_t = x_{t,i} \le 1/4$.

Let us now trace our steps back from $t$ to $0$. We claim that all transitions before time $t$ were made by agent $i$, i.e., for every $\tau < t$, $i = i_{\tau}$. If not, let $\tau < t$ be the most recent transition by an agent $k \neq i$. As $x_{\tau+1,k} = x_{t,k} = 0$, therefore $s_{\tau} - x_{\tau,k} \ge 1$, but $s_{\tau} - x_{\tau,k} = s_{t} - x_{t,k} = x_{t,i} \le 1/4$. Contradiction. 

As $i$ makes all the transitions before time $t$, we have either (i) if $t = 0$, then $x_{t,i} = x_{0,i}$; or (ii) if $t > 0$, then $x_{t,i} = a$ as a response to $0$ output by everyone else. So, $s_{t+1} \ge x_{t+1,i} = x_{t,i} \ge \min(a, x_{0,i})$.

\noindent\textbf{Case 2:} Definition~\ref{def:warmup} condition \eqref{def:warmup:3} violated at $t$, i.e., total output $s_t \ge 1$.

Let $i$ be the agent that makes the move at $t$ to decrease total output from $s_t \ge 1$ to $s_t < 1$. We argued in Case 1 that condition \eqref{def:warmup:2} must have been satisfied at $t$. So, there are at least two agents with strictly positive output at $t$. Let $j \neq i$ be an agent other than $i$ that has $x_{t,j} > 0$. We claim that $x_{t,j} = x_{0,j}$.

We prove our claim by contradiction. Let us trace our steps back from $t$. Let agent $k$ make the transition at $t-1$. We show that $k \neq i, j$ because:
\begin{itemize}
    \item $k \neq i$ because: If $k = i$, then the transition at $t$ would have been redundant and $1 > s_{t+1} = s_t \ge 1$, contradiction.
    \item $k \neq j$ because: Notice that $x_{t,k} = 0$ because if $x_{t,k} > 0$ then either: (i) $s_{t-1} - x_{t-1,k} = 0$, but this is not possible as $s_{t-1} - x_{t-1,k} \ge x_{t-1,i} = x_{t,i} > 0$; (ii) $0 < \sqrt{s_{t-1} - x_{t-1,k}} < 1$, but this is also not possible because then $1 \le s_t = \sqrt{s_{t-1} - x_{t-1,k}} < 1$. As $x_{t,j} > 0$, so $k \neq j$.
\end{itemize}
Repeating the same argument, for every $\tau < t$, we can show that $i_{\tau} \neq j$. So, $x_{t,j} = x_{0,j} \implies s_{t+1} \ge x_{t+1,j} = x_{t,j} = x_{0,j}$.

\noindent\textbf{Case 3:} Definition~\ref{def:warmup} condition \eqref{def:warmup:1} violated at $t$, i.e., there is an agent $i$ with $x_{t,i} > 1/4$.

Let $\alpha = \sum_{j \neq i} x_{t,j}$. W.l.o.g., we can assume that conditions \eqref{def:warmup:2} and \eqref{def:warmup:3} are satisfied at time $t$ because we have already considered their violation in the previous cases. So, $\alpha > 0$ (by condition \eqref{def:warmup:2}) and $s_t = x_{t,i} + \alpha < 1 \implies \alpha < 1 - x_{t,i} < 1$ (by condition \eqref{def:warmup:3}). As $\alpha < 1$, we have $s_{t+1} = \sqrt{s_t - x_{t,i}} = \sqrt{\alpha}$. Let us now lower bound $\alpha$.

If $t=0$, then we trivially have $s_{t+1} \ge \sqrt{\alpha} \ge \alpha = \sum_{j \neq 1} x_{0,j}$ (and we know that $\alpha > 0$, so there is an agent $j \neq i$ with positive $x_{0,j}$).
Now, let $t > 0$.
In the proof of Lemma~\ref{lm:part1}, we argued that if an agent $j$ plays at least one move before time $\tau$, for any $\tau$, then $x_{\tau,j} \le 1/4$. As $x_{t,i} > 1/4$, so $i_{\tau} \neq i$ for all $\tau < t$, which implies $x_{t,i} = x_{0,i}$. 
As $\alpha > 0$, $x_{t,i} > 0$ and $s_t = \alpha + x_{t,i} < 1$, we know that $s_{t} = \sqrt{s_{t-1} - x_{t-1,i_{t-1}}} = \sqrt{x_{t-1,i} + \beta} = \sqrt{x_{0,i} + \beta}$ for some $\beta \in [0, 1-x_{0,i})$. So, $s_t = \alpha + x_{t,i} = \alpha + x_{0,i} \ge \min_{\beta \in [0, 1-x_{0,i})} \sqrt{x_{0,i} + \beta} = \sqrt{x_{0,i}}$, which implies that $\alpha \ge \sqrt{x_{0,i}}(1 - \sqrt{x_{0,i}}) > (1 - \sqrt{x_{0,i}})/2$ as $x_{0,i} > 1/4$. So, $s_{t+1} = \sqrt{\alpha} \ge \sqrt{(1 - \sqrt{x_{0,i}})/2}$.
\end{proof}

%% file: proofs/markov.tex
\begin{proof} [Proof of Lemma~\ref{lm:markov}]
Consider the biased random walk $(z_t)_{t \ge 0}$, where $z_{t+1} = z_{t} + 1$ w.p. $p < 1/2$ and $z_{t+1} = z_{t} - 1$ w.p. $1-p > 1/2$. Notice that the behavior of $z_t$ is almost similar to the behavior of $y_t$ defined earlier (Figure~\ref{fig:markov}), except that $y_t$ cannot move left from state $1$ but $z_t$ can. Let us assume that $z_0 = y_0 = k$ and $z_t$ is coupled with $y_t$, i.e., $z_t$ moves left if $y_t$ moves left (or stays at state $1$) and $z_t$ moves right if $y_t$ moves right. Notice that $y_t \ge z_t$, for all $t$, by construction.

It is easy to check that, if $z_{\tau} \le 1$ at time $t$, then $y_{t} = 1$ for some $t \le \tau$. Because, if $y_{\tau} = z_{\tau}$, then trivially $y_t = 1$ for $t = \tau$. Else $y_{\tau} > z_{\tau}$, which can only happen if $z_{t} = y_{t} = 1$ at some time $t$ and $z_{t+1} = 0$ but $y_{t+1} = 1$. Similarly, we can show that, if $z_{\tau} \le 2 - m$, then $y_{t}$ has visited state $1$ at least $m$ times by time $\tau$, i.e., $| \{ t \le \tau \mid y_{t} = 1 \} | \ge m$. The argument is similar as before: The first time $z_{t}$ is at $1$, $y_t$ is also at $1$ (the first visit). Further, the difference between $y_t$ and $z_t$ increases only when $y_t$ is at state $1$ and tries unsuccessfully to move left but $z_t$ moves left. As $y_{\tau} \ge 1$ and $z_{\tau} \le 2 - m$, so $y_{\tau} - z_{\tau} \ge m - 1$, so $y_t$ makes at least $m-1$ additional visits to state $1$ by time $\tau$. 

Next, we bound the time it takes for $z_t$ fall below $2 - m$. Let $Z_i$ be a random variable that takes value $1$ w.p. $p$ and $-1$ w.p. $1-p$. $z_t$ can be written as $z_t = \sum_{i = 1}^{t} Z_i + k$. Notice that $\Exp[\sum_{i = 1}^{t} Z_i] = t(2p - 1)$. Applying Hoeffding's inequality, we have
\begin{align*}
    \Prob[z_t > 2 - m] &= \Prob\left[ \sum_{i = 1}^t Z_i + k \ge 3 - m \right] = \Prob\left[ \sum_{i = 1}^t Z_i - \Exp\left[\sum_{i = 1}^t Z_i \right] \ge 3 - m - k + t(1 - 2p) \right] \\
    &\le \exp\left( \frac{-2 (t(1 - 2p) - m - k + 3)^2}{4 t } \right) \le \exp\left( \frac{-(t(1 - 2p) - m - k + 3)^2}{2 t } \right).
\end{align*}
Setting $t \ge \frac{4}{1-2p} \max\left( m + k, \frac{1}{1-2p} \ln\left( \frac{1}{\delta} \right) \right)$, we get
\[
    \exp\left( \frac{-(t(1 - 2p) - m - k + 3)^2}{2t } \right) \le \exp\left( \frac{- 3^2 t^2 (1 - 2p)^2 }{4^2 2 t } \right) \le \exp\left( \frac{-9}{8} \frac{t (1 - 2p)^2 }{4} \right) \le \delta.
\]
\end{proof}

%% file: proofs/potential.tex
\begin{proof}[Proof of Lemma~\ref{lm:potential}]
Proof of Lemma~\ref{lm:potential} \eqref{lm:potential:1}.
As $z_i^* = \frac{n-1}{n^2}$, we have $\sigma^* = \frac{n-1}{n}$ and $\sum_{i < j} z^*_i z^*_j = \binom{n}{2} \left(\frac{n-1}{n^2}\right)^2 = \frac{(n-1)^3}{2 n^3}$. Plugging this into $f$ we get
\begin{equation*}\label{eq:lm:potential:1}
    f(z^*) = \frac{1}{3} (\sigma^*)^3 - \sum_{i < j} z^*_i z^*_j + \frac{1}{6} \left(\frac{n-1}{n} \right)^3 = \frac{1}{3} \left(\frac{n-1}{n} \right)^3 - \binom{n}{2} \left(\frac{n-1}{n^2} \right)^2 + \frac{1}{6} \left(\frac{n-1}{n} \right)^3 = 0.
\end{equation*}

Proof of Lemma~\ref{lm:potential} \eqref{lm:potential:2}.
Notice that we can write $\sum_{i < j} z_i z_j$ as $ \left( \sum_{i} z_i \right)^2 - \sum_{i} z_i^2$. Using this
\begin{equation*}
    f(z) = \frac{1}{3} \sigma^3 - \sum_{i < j} z_i z_j + \frac{1}{6} \left(\frac{n-1}{n} \right)^3 = \frac{1}{3} \sigma^3 - \frac{1}{2} \sigma^2 + \frac{1}{2} \sum_{i} z_i^2 + \frac{1}{6} \left(\frac{n-1}{n} \right)^3.
\end{equation*}
Now, as $\sum_{i} z_i^2$ is a convex function of $z$, we have 
\[
    \sum_{i} z_i^2 \ge \sum_{i} \left( \frac{\sigma}{n} \right)^2 = \frac{\sigma^2}{n} \implies f(z) \ge \frac{1}{3} \sigma^3 - \left(\frac{n-1}{2n} \right) \sigma^2 + \frac{1}{6} \left(\frac{n-1}{n} \right)^3.
\]
Let $g(y) = \frac{1}{3} y^3 - \frac{n-1}{2n} y^2$ for $y \in [0,1]$. Notice that $f(z) \ge g(\sigma) - \frac{1}{6} \left(\frac{n-1}{n} \right)^3$. Let us now lower bound $g(y)$.
\[
    g'(y) = y^2 - \frac{n-1}{n} y = 0 \implies y \in \left\{ 0, \frac{n-1}{n} \right\}.
\]
Notice that $g'(y) \le 0$ if $y \in [0, (n-1)/n]$ and $g'(y) \ge 0$ for $y \ge (n-1)/n$. So,
\begin{multline*}
    \min_{y \in [0,1]} g(y) 
    = g\left(\frac{n-1}{n}\right) 
    = \frac{1}{3} \left(\frac{n-1}{n}\right)^3 - \frac{n-1}{2n} \left(\frac{n-1}{n}\right)^2 
    = - \frac{1}{6} \left(\frac{n-1}{n}\right)^3  \\
    \implies f(z) 
    \ge g(\sigma) + \frac{1}{6} \left(\frac{n-1}{n} \right)^3 
    \ge g\left(\frac{n-1}{n}\right)  + \frac{1}{6} \left(\frac{n-1}{n} \right)^3 
    = 0.
\end{multline*}
On the other hand, $f(z) \le \frac{1}{2}$ because $f(z) = \frac{1}{3} \sigma^3 - \sum_{i < j} z_{i} z_{j} + \frac{1}{6}\left(\frac{n-1}{n}\right)^3 \le \frac{1}{3} \sigma^3 + \frac{1}{6} \le \frac{1}{2}$.

Proof of Lemma~\ref{lm:potential} \eqref{lm:potential:3}. 
We are given that $\sigma \le 3(n-1)/(4n)$. We showed earlier that $f(z) \ge g(\sigma) + ((n-1)/n)^3/6$ and that $g(y)$ decreases for $y \in [0, (n-1)/n]$, so
\begin{multline*}
    \min_{y \in [0, \frac{3(n-1)}{4n}]} g(y) \ge g\left(\frac{3(n-1)}{4n}\right) = \left( \frac{1}{3}\frac{3^3}{4^3} - \frac{1}{2}\frac{3^2}{4^2} \right) \left(\frac{n-1}{n} \right)^3 = -\frac{3^2}{4^3} \left(\frac{n-1}{n} \right)^3 \implies \\
    f(z) \ge g(\sigma) + \frac{1}{6} \left(\frac{n-1}{n} \right)^3 \ge g\left(\frac{3(n-1)}{4n}\right)  + \frac{1}{6} \left(\frac{n-1}{n} \right)^3 =  \left(\frac{1}{6} - \frac{3^2}{4^3}\right) \left(\frac{n-1}{n} \right)^3 \ge \frac{1}{40} \left(\frac{n-1}{n} \right)^3.
\end{multline*}

Proof of Lemma~\ref{lm:potential} \eqref{lm:potential:4}. 
We are given that $\sigma \ge 3(n-1)/(4n)$. Let's show that $f$ is $1$-strongly convex if $\sigma \ge 3(n-1)/(4n)$. Taking derivative of $f$ we get
\begin{equation*}
    \frac{\partial f}{\partial z_i} = \sigma^2 - \sum_{j \neq i} z_j, \qquad
    \frac{\partial^2 f}{\partial z_i^2} = 2\sigma, \qquad 
    \frac{\partial^2 f}{\partial z_i \partial z_j} = 2\sigma - 1, \qquad \text{ for $i,j \in [n]$, $i \neq j$.}
\end{equation*}
So, the hessian $\nabla^2 f$ has $(\nabla^2 f)_{i,j} = 2 \sigma$ if $i = j$ and $(\nabla^2 f)_{i,j} = 2 \sigma - 1$ if $i \neq j$. The hessian can be written as $\nabla^2 f = (2 \sigma - 1) + I$, where $I$ is the $n \times n$ identity matrix. Let us now find the eigenvalues of $\nabla^2 f$; $\lambda$ is an eigenvalue if  $\det(\nabla^2 f - \lambda I) = \det((2 \sigma - 1) + (1-\lambda) I) = 0$. Using matrix determinant lemma,\footnote{\url{https://en.wikipedia.org/wiki/Matrix_determinant_lemma}} we have
\[
    \det((2 \sigma - 1) + (1-\lambda) I) = \left( 1 + \frac{(2\sigma-1)n}{1-\lambda} \right) (1-\lambda)^n = (1-\lambda)^{n-1} (2n\sigma - (n-1) - \lambda) = 0,
\]
which gives us $\lambda = 1$ or $\lambda = 2n\sigma - (n-1) = 2n (\sigma - \frac{n-1}{2n}) \ge 2n (\frac{3(n-1)}{4n} - \frac{n-1}{2n}) = \frac{n-1}{2} \ge 1$. As all eigenvalues of the hessian $\nabla^2 f$ are at least $1$, so $f$ is $1$-strongly convex. Now, as $f$ is $1$-strongly convex, we have
\[
    f(z) \ge f(z^*) + (\nabla f(z^*))^{\intercal} (z - z^*) + \frac{1}{2} || z - z^* ||_2^2.
\]
We have shown earlier that $f(z^*) = 0$. Further, as
\[
    (f(z^*))_i = \left. \frac{\partial f}{\partial z_i} \right|_{z^*} = (\sigma^*)^2 - \sum_{j \neq i} z^*_j = \left( \frac{n-1}{n} \right)^2 - \sum_{j \neq i} \left( \frac{n-1}{n^2} \right) = 0,
\]
for every $i$, we have $\nabla f(z^*) = 0$, which gives our required result $f(z) \ge \frac{1}{2} || z - z^* ||_2^2$.

Proof of Lemma~\ref{lm:potential} \eqref{lm:potential:5}.
The proof is similar to that of Lemma~\ref{lm:potential} \eqref{lm:potential:4}, but instead of strong convexity, we focus on Lipschitz smoothness.

Let $y_i = \sum_{j \neq i} z_i = \sigma - z_i$. Let $\rho = \sum_i y_i = \sum_i (\sigma - z_i) = (n-1) \sigma \le (n-1)$. Further,
\begin{multline*}
    \sum_{i < j} z_i z_j = \sum_{i < j} (\sigma - y_i) (\sigma - y_j) = \binom{n}{2} \sigma^2 - \sigma \sum_{i < j} (y_i + y_j) + \sum_{i < j} y_i y_j \\
    = \frac{n (n-1)}{2} \sigma^2 - \sigma (n-1) \rho + \sum_{i < j} y_i y_j = \left( \frac{n}{2(n-1)} - 1 \right) \rho^2 + \sum_{i < j} y_i y_j =  - \frac{n - 2}{2(n-1)} \rho^2 + \sum_{i < j} y_i y_j.
\end{multline*}
Let us define the function $h(y)$ such that $h(y) = f(z)$; plugging in $\sigma = \frac{1}{(n-1)} \rho$ and $\sum_{i < j} z_i z_j = - \frac{n - 2}{2(n-1)} \rho^2 + \sum_{i < j} y_i y_j$, we have
\[
    h(y) = f(z) = \frac{1}{3 (n-1)^3} \rho^3 + \frac{n - 2}{2(n-1)} \rho^2 - \sum_{i < j} y_i y_j + \frac{1}{6} \left(\frac{n-1}{n} \right)^3 .
\]
Taking derivative of $h$ w.r.t. $y$ we get
\begin{align*}
    \frac{\partial h}{\partial y_i} &= \frac{1}{(n-1)^3} \rho^2 + \frac{n-2}{n-1} \rho - (\rho - y_i) = \frac{1}{(n-1)^3} \rho^2 - \frac{1}{n-1} \rho + y_i, \\
    \frac{\partial^2 h}{\partial y_i^2} &= \frac{2}{(n-1)^3} \rho - \frac{1}{n-1} \rho + 1,\qquad \frac{\partial^2 h}{\partial y_i \partial y_j} = \frac{2}{(n-1)^3} \rho - \frac{1}{n-1}, \qquad \text{ for $i,j \in [n]$, $i \neq j$.}
\end{align*}
So, the hessian can be written as $\nabla^2 h = \left( \frac{2}{(n-1)^3} \rho - \frac{1}{n-1} \right) + I$, where $I$ is the $n \times n$ identity matrix. Let us now find the eigenvalues of $\nabla^2 h$. Using the matrix determinant lemma, we have
\begin{multline*}
    \det\left( \left(\frac{2}{(n-1)^3} \rho - \frac{1}{n-1} \right) + (1-\lambda) I \right) = \left( 1 + \frac{n \left( \frac{2}{(n-1)^3} \rho - \frac{1}{n-1} \right)}{1-\lambda} \right) (1-\lambda)^n \\
    = (1-\lambda)^{n-1}  \left( \frac{2n}{(n-1)^3} \rho - \frac{n}{n-1} + 1 - \lambda \right),
\end{multline*}
which gives us $\lambda = 1$ or $\lambda = \frac{2n}{(n-1)^3} \rho - \frac{1}{n-1} \le \frac{2n(n-1)}{(n-1)^3} - \frac{1}{n-1} = \frac{2n - (n-1)}{(n-1)^2} = \frac{n+1}{(n-1)^2} = \frac{1}{(n-1)} + \frac{2}{(n-1)^2} \le 1$. As all eigenvalues of the hessian $\nabla^2 h$ are at most $1$, so $h(y)$ is $1$-Lipschitz smooth w.r.t. $y$. Let $y^*_i = \sum_{j \neq i} z_i^* = \left(\frac{n-1}{n}\right)^2$ and $\rho^* = \sum_i y_i^* = \frac{(n-1)^2}{n}$. Notice that $h(y^*) = f(z^*) = 0$ and $\nabla h (y^*) = 0$ because
\[
    (h(y^*))_i = \left. \frac{\partial h}{\partial y_i} \right|_{y^*} = \frac{1}{(n-1)^3} \frac{(n-1)^4}{n^2} - \frac{1}{n-1} \frac{(n-1)^2}{n} + \frac{(n-1)^2}{n^2} = \frac{n-1}{n^2} ( 1 - n + (n-1) ) = 0,
\]
for every $i \in [n]$. This gives us
\[
    f(z) = h(y) \le h(y^*) + (\nabla h(y^*))^{\intercal} (y - y^*) + \frac{1}{2} || y - y^* ||_2^2 = \frac{1}{2} \sum_i (y_i - y_i^*)^2 = \frac{1}{2} \sum_i \left( \sum_{j \neq i} ( z_i - z_i^* ) \right)^2 .
\]
\end{proof}

%% file: proofs/lipschitz.tex
\begin{proof}[Proof of Lemma~\ref{lm:lipschitz}]
 As $|| z - z^* ||_2 \le \epsilon$, we have
\begin{itemize}
    \item $| z_i - z_i^* | \le || z - z^* ||_2 \le \epsilon$ for every $i \in [n]$,
    \item using Cauchy--Schwarz inequality, $| \sigma - \sigma^* | = | \sum_i (z_i - z_i^*) | \le \sqrt{n \sum_i (z_i - z_i^*)^2 } = \sqrt{n} || z - z^* ||_2 \le \sqrt{n} \epsilon $,
    \item again using Cauchy--Schwarz inequality, $| (\sigma - z_i) - (\sigma^* - z_i^*) | = | \sum_{j \neq i} (z_j - z_j^*) | \le \sqrt{n-1}\epsilon$ for every $i \in [n]$.
\end{itemize}
Remember that $z_i^* = \frac{n-1}{n^2}$, $\sigma^* = \sum_j z_j^* = \frac{n-1}{n}$, and $\sigma^* - z_i^* = \sum_{j \neq i} z_j^* = \left(\frac{n-1}{n}\right)^2$ for every $i \in [n]$.
 
Let us fix an arbitrary agent $i$.
Let us prove that the best response of agent $i$ is $(1/2)$-Lipschitz smooth with respect to the output of other agents, near the equilibrium point. We can write the following bounds on $\sigma - z_i$
\begin{itemize}
    \item $\sigma - z_i \le (\sigma^* - z^*_i) + \sqrt{n-1} \epsilon < \left(\frac{n-1}{n}\right)^2 + \frac{\sqrt{n-1}}{n \sqrt{n}} < \left(\frac{n-1}{n}\right) + \frac{1}{n} \le 1$,
    \item $\sigma - z_i \ge (\sigma^* - z^*_i) - \sqrt{n-1} \epsilon > \left(\frac{n-1}{n}\right)^2 - \frac{\sqrt{n-1}}{n \sqrt{n}} \ge \left(\frac{n-1}{n}\right)^2 - \frac{(n-1)\sqrt{n-1}}{2n \sqrt{n}} = \left(\frac{n-1}{n}\right)^{3/2} \left( \sqrt{\frac{n-1}{n}} - \frac{1}{2} \right) \ge \left(\frac{2}{3}\right)^{3/2} \left( \sqrt{\frac{2}{3}} - \frac{1}{2} \right) \ge \frac{1}{6}$.
\end{itemize}
Let $g(y) = \sqrt{y} - y$. The best response of agent $i$ is $g(y) = \sqrt{y} - y$ given the total output by others is $y$. We claim $|g'(y)| \le 1/2$ for $y \in [1/9, 1]$ because (i) $g'(y) = \frac{1}{2 \sqrt{y}} - 1 \ge -1/2$ for $y \le 1$, and (ii) $g'(y) = \frac{1}{2 \sqrt{y}} - 1 \le 1/2$ for $y \ge \frac{1}{9}$.

Let $\widehat{z}_i$ denote the best response of agent $i$ given the total output of other agents is $\sigma_i - z_{i}$. As $\sigma_i - z_{i} \ge 1/6$ and $\sigma^*_i - z^*_i = ((n-1)/n)^2 \ge 4/9$, and the best response is $1/2$-smooth in $[1/9,1]$: 
\[
    |\widehat{z}_i - z_i| \le |\widehat{z}_i - z^*_i| + |z^*_i - z_i| \le \frac{1}{2} | (\sigma_i - z_i) - (\sigma^*_i - z^*_i) | + |z^*_i - z_i| \le \left(\frac{\sqrt{n-1}}{2} + 1\right) \epsilon \le \sqrt{n} \epsilon.
\]

Let us now prove that the utility of agent $i$ is $3$-Lipschitz smooth with respect to his output, near the equilibrium profile. Using $1/2$-smoothness of best response, we have $\widehat{z}_i - z^*_i \le \frac{1}{2} | (\sigma_i - z_i) - (\sigma^*_i - z^*_i) | \le \frac{\sqrt{n-1} \epsilon}{2}$. We can write the following bounds on $(\sigma - z_i) + \widehat{z}_i$: 
\begin{itemize}
    \item $(\sigma - z_i) + \widehat{z}_i = \sqrt{\sigma - z_i} < 1$ because $\widehat{z}_i$ is the best response to $(\sigma - z_i)$ and $\sigma - z_i < 1$ (shown earlier),
    \item $(\sigma - z_i) + \widehat{z}_i \ge ((\sigma^* - z^*_i) - \sqrt{n-1} \epsilon) + (z^*_i - \frac{\sqrt{n-1} \epsilon}{2}) > \frac{n-1}{n} - \frac{3 \sqrt{n-1}}{2 n \sqrt{n}} \ge \frac{2}{3} - \frac{3 \sqrt{2}}{2 \cdot 3 \sqrt{3}} \ge \frac{1}{4}$.
\end{itemize}
Similarly, we can write the following bounds on $\sigma$:
\begin{itemize}
    \item $\sigma \le \sigma^* + \sqrt{n} \epsilon < \frac{n-1}{n} + \sqrt{n}\frac{1}{n \sqrt{n}} = 1$,
    \item $\sigma \ge \sigma^* - \sqrt{n} \epsilon > \frac{n-1}{n} -\sqrt{n}\frac{1}{n \sqrt{n}} = \frac{n-2}{n} \ge \frac{1}{3}$.
\end{itemize}

Let $h(y) = u_i(y, z_{-i}) = \frac{y}{(\sigma - z_i) + y} - y$, we have $h'(y) = \frac{\partial u_i(y, z_{-i})}{\partial y} = \frac{\sigma - z_i}{(\sigma - z_i + y)^2} - 1$. We claim $|h'(y)| \le 3$ for $(\sigma - z_i) + y \ge 1/4$ because (i) $h'(y) = \frac{\sigma - z_i}{(\sigma - z_i + y)^2} - 1 \ge -1$, and (ii) $h'(y) \le \frac{\sigma - z_i}{(\sigma - z_i + y)^2} - 1 \le \frac{1}{\sigma - z_i + y} - 1 \le 3$. Now, as $(\sigma - z_i) + \widehat{z}_i \ge 1/4$ and $(\sigma - z_i) + z_i = \sigma \ge 1/3$, we can use this $3$-smoothness of the utility function to get $|u_i(\widehat{z}_i, z_{-i}) - u_i(z)| \le 3 |\widehat{z}_i - z_i| \le 3 \sqrt{n} \epsilon$.
\end{proof}